\documentclass[a4paper,12pt]{article}
\pdfoutput=1
\usepackage{jcappub}
\usepackage[utf8]{inputenc}
\usepackage[table,dvipsnames]{xcolor}
\usepackage{multirow}
\usepackage{graphicx}
\usepackage{mathtools}
\usepackage{amsmath,amssymb,amsfonts,amsthm,latexsym,stmaryrd,physics,mathrsfs}
\allowdisplaybreaks[4]
\usepackage{marginnote}
\usepackage{color}
\usepackage{bm}
\usepackage{longtable}

\usepackage{float}
\usepackage{nicefrac}
\usepackage{microtype} 
\usepackage{booktabs}
\usepackage{verbatim}
\usepackage{setspace}
\usepackage[T1]{fontenc}
\usepackage[utf8]{inputenc}
\usepackage{stfloats}
\usepackage{diagbox}
\usepackage{dsfont}
\usepackage[colorinlistoftodos,prependcaption,textsize=tiny]{todonotes}
\usepackage{epstopdf}
\usepackage{latexsym}
\usepackage{xparse}
\usepackage[center]{subfigure}

\usepackage{xargs}                      %
\usepackage[colorinlistoftodos,prependcaption,textsize=tiny]{todonotes}

\usepackage{ulem} %

 \newcommand{\bq}{\begin{equation}}
 \newcommand{\eq}{\end{equation}}
 \newcommand{\bqn}{\begin{eqnarray}}
 \newcommand{\eqn}{\end{eqnarray}}

\NewDocumentCommand{\evalat}{sO{\big}mm}{%
  \IfBooleanTF{#1}
   {\mleft. #3 \mright|_{#4}}
   {#3#2|_{#4}}%
}

\newtheorem{definition}{Definition}

\newtheorem{lemma}{Lemma}

\newtheorem{corollary}{Corollary}
\def\be{\begin{eqnarray}}
\def\ee{\end{eqnarray}}

\renewcommand{\a}{\alpha}

\newcommand{\ltbf}{\Xi}
\newtheorem*{remark}{Remark}

\title{Embedding generalized LTB models in polymerized spherically symmetric spacetimes}

\author[1]{Kristina Giesel,}
\emailAdd{kristina.giesel@gravity.fau.de} 
\affiliation[1]{Department Physik, Institut f\"ur Quantengravitation, Theoretische Physik III, Friedrich-Alexander Universit\"at Erlangen-N\"urnberg, Staudtstr. 7/B2, 91058 Erlangen, Germany}

\author[1]{Hongguang Liu,} 
\emailAdd{hongguang.liu@gravity.fau.de}

\author[1]{Eric Rullit,} 
\emailAdd{eric.rullit@fau.de}

\author[2]{\\ Parampreet Singh,}
\emailAdd{psingh@lsu.edu}
\affiliation[2]{Department of Physics and Astronomy, Louisiana State University, Baton Rouge, LA 70803, USA}

\author[1]{Stefan Andreas Weigl} 
\emailAdd{stefan.weigl@gravity.fau.de}

\abstract{We generalize the existing works on the way (generalized) LTB models can be embedded into polymerized spherically symmetric models in several aspects. We re-examine such an embedding at the classical level and show that a suitable LTB condition can only be treated as a gauge fixing in the non-marginally bound case, while in the marginally bound case it must be considered as an additional first class constraint. A novel aspect of our formalism, based on the effective equations of motion, is to derive compatible dynamics LTB conditions for polymerized models by using holonomy and inverse triad corrections simultaneously, whereas in earlier work these were only considered separately. Further, our formalism allows to derive compatible LTB conditions for a vast of class of polymerized models available in the current literature. Within this broader class of polymerizations there are effective models contained for which the classical LTB condition is a compatible one. Our results show that there exist a class of effective models for which the dynamics decouples completely along the radial direction. It turns out that this subsector is strongly linked to the property that in the temporally gauge fixed model, the algebra of the geometric contribution to the Hamiltonian constraint and the spatial diffeomorphism constraint is closed.  We finally apply the formalism to existing models from the literature and compare our results to the existing ones.}

\notoc
\begin{document}
\maketitle

\newpage
\begin{table}[ht]
\begin{center} 
    \begin{longtable}{  m{6cm}| m{11cm}  } 
\toprule
\rowcolor{gray!10}  \multicolumn{2}{c}{LIST OF NOTATION} \\
\toprule
 $\ltbf(x)$ & LTB function, the marginally bound case is defined as $\frac{d}{dx}\ltbf(x)=\ltbf^\prime(x)=0$.\\
        \hline
        $f(K_x, K_\phi, E^x, E^\phi)$ &  polymerization function of the term involving extrinsic curvature variables, see eq.  \eqref{eq:defpolyhamiltonianconstraint} for the definition, can also includes inverse triad corrections. $f\to 0$ in the classical limit.\\
        \hline
        $f^{(1)}(K_\phi, E^x), f^{(2)}(K_\phi, E^x)$ & special form of $f$ appearing in the case of a closed constraint algebra  (Lemma \ref{lemma:f1_CC}, \eqref{eq:no_kx_f1}) whose classical limits read $f^{(1)} \to K_{\phi}^2$ and $f^{(2)} \to K_{\phi}$.\\
        \hline
        $h_1(E^x), h_2(E^x)$& inverse triad corrections of the remaining terms involving $\frac{1}{\sqrt{E^x}}$, $\frac{E^x}{\sqrt{E^x}}$ terms, see \eqref{eq:defpolyhamiltonianconstraint}.\\
        \hline
        $\widetilde{K}_x$ & density weight zero combination of the canonical variables $\widetilde{K}_x \coloneqq \frac{K_x}{E^\phi}$ used to parameterize $g_\Delta$, see Lemma \ref{lemma:g1_first}. \\
        \hline
        $g$ & combination of canonical variables  $g=\frac{ \partial_x {E^x}}{2 E^{\phi}}$ used to derive an alternative set of effective EOM (\ref{eq:Kp_ltb_eom}-\ref{eq:ltb_eom}) on the constraint hypersurface $C_x = 0$.\\
        \hline
        $g_c=\ltbf$ & classical LTB condition. \\
        \hline
        $g_\Delta$ & generalized form of the classical LTB condition $g_c = \ltbf$, to study all possible effective LTB conditions, see definition \ref{def:comp_LTB}. \\
        \hline
        $g^{(1)}_\Delta( K_{\phi}, E^x, \ltbf), g^{(2)}_\Delta(\widetilde{K}_x, K_{\phi}, E^x)$ & two classes of solutions for $g_\Delta$ corresponding to the marginally and non-marginally bound case, see Lemma \ref{lemma:g1_first}, \ref{thm:marginal}, \ref{thm:non_marginally bound case }. \\
        \hline
        $\tilde{g}_\Delta$ & special class of solutions such that $g^{(1)}_\Delta( K_{\phi}, E^x, \ltbf) = \tilde{g}_\Delta( K_{\phi}, E^x)\ltbf$, see Corollary \ref{cor:gEx_non_m}.\\
 \bottomrule       
    \end{longtable}
\end{center}
\caption{List of notation used in this article.}
\label{tab:Notation}
\end{table}
\newpage

\section{Introduction}
\label{sec:intro}
The investigation of symmetry reduced models in the context of loop quantum gravity (LQG) has gained much interest in recent years, in particular with focus on cosmological, see \cite{Li:2023dwy} for a recent review on physical implications of loop quantum cosmology (LQC), as well as spherically symmetric models \cite{Ashtekar:2005qt,Modesto:2005zm,Boehmer:2007ket,Chiou:2012pg,Gambini:2013hna,Brahma:2014gca,Dadhich:2015ora,Tibrewala:2013kba,BenAchour:2017ivq,Yonika:2017qgo,DAmbrosio:2020mut,Olmedo:2017lvt,Ashtekar:2018lag,Ashtekar:2018cay,Bojowald:2018xxu,BenAchour:2018khr,Bodendorfer:2019cyv,Alesci:2019pbs,Assanioussi:2019twp,Benitez:2020szx,Gan:2020dkb,Gambini:2020qhx,Husain:2021ojz,Husain:2022gwp,Li:2021snn,Gan:2022mle,Kelly:2020lec,Gambini:2020nsf,Han:2020uhb,Zhang:2021xoa,Munch:2022teq,Lewandowski:2022zce} or \cite{Ashtekar:2023cod} for a recent review on these models. While for both spacetimes symmetry reduced quantum models exist \cite{Ashtekar:2006rx,Ashtekar:2006uz,Ashtekar:2006wn,Bojowald:2004af,Bojowald:2004ag,Bojowald:2005cb,Corichi:2015xia,Gambini:2022hxr}, in general their quantum dynamics are still too complicated to cope with and therefore 
a lot of work in the literature has focused on so-called effective models for these spacetimes in the framework of LQC and loop quantum black holes respectively. A common feature of these effective models is that the holonomy operators of full LQG are replaced by so-called polymerized functions of connection and inverse volume operators by so-called inverse triad corrections. These two modifications capture the quantum geometric effects from full LQG in the loop inspired quantization of these symmetry reduced models with the main result that singularities in these spacetimes are resolved and in the cosmological models one finds a big bounce replacing the big bang \cite{Ashtekar:2006rx,Ashtekar:2006wn,Ashtekar:2011ni}.
\footnote{Here we should note that while the inverse volume corrections are generally found to be of little consequence for singularity resolution in models which are non-compact or without an intrinsic curvature \cite{Ashtekar:2006wn}, they both descend from LQG and the inverse volume modifications can be potentially important in models with an intrinsic curvature. See for eg. \cite{Motaharfar:2022pjp,Motaharfar:2023gpp} where their role in the initial conditions near the to-be classical singularity is discussed.}  In this sense the effective models can be understood as an intermediate step between the classical and the full quantum model in which already some of the features the quantum model is expected to have can be investigated in a simpler setup. Recently also effective Lema\^itre-Tolman-Bondi (LTB) models have been analyzed to describe the quantum gravity dust collapse \cite{Bojowald:2008ja,Bojowald:2009ih,Bambi:2013caa,Kelly:2020lec, BenAchour:2020bdt, Munch:2020czs,Husain:2021ojz,Husain:2022gwp, Giesel:2022rxi,Bobula:2023kbo,Fazzini:2023scu}. At the classical level LTB models \cite{lematre,tolman,bondi} can be embedded into a spherically symmetric spacetime with dust by solving the spatial diffeomorphism constraint together with an additional so-called LTB condition. This imposes further conditions on the independent variables of spherically symmetric spacetimes yielding a simplified set of dynamical equations and removes some of the gauge freedom involved. Compared to the classical spherically symmetric models, LTB models are also attractive because of their simplifications, as the corresponding quantization or effective models are easier to handle.
~\\
~\\
In the context of effective models, an important question to ask is to what extent the reduction from spherically symmetric to LTB models, which is possible in the classical theory, can also be carried out at the level of effective models when quantum geometric corrections are included. In a broader context, this is related to the question to what extent the classical reduction from spherically symmetric model to LTB and the subsequent quantization commutes with first quantizing the spherically symmetric model and then performing the reduction to the LTB sector. A seminal analysis can be found in \cite{Bojowald:2008ja} for the marginally bound LTB model and for the non-marginally bound model in \cite{Bojowald:2009ih} in the framework of Ashtekar-Barbero variables. The strategy followed in both cases is to formulate a so-called LTB condition which relates the two triad variables present in spherically symmetric models and require it to be stable under the effective dynamics. For the models considered in \cite{Bojowald:2008ja,Bojowald:2009ih} this necessarily leads to a modification of the classical LTB condition and in turn to a modified LTB metric. 
~\\
~\\
While these models have been studied in the LQG community for over an decade, our  manuscript aims to revisit the assumptions and analysis in \cite{Bojowald:2008ja,Bojowald:2009ih}  and improve and generalize it 
in several aspects: First we will re-examine the classical case and show that, in contrast to \cite{Bojowald:2008ja}, for the marginally bound case the LTB condition cannot be treated as a gauge-fixing condition but needs to implemented as an additional first class constraint. Second, for the non-marginally bound case we show that the corresponding Dirac bracket algebra selects a preferred pair of canonical variables as far as the latter quantization of the model is concerned. Third, we move on to effective models with a much broader class of polymerizations than the work in \cite{Bojowald:2008ja,Bojowald:2009ih}. Given the corresponding effective dynamics of a chosen polymerization, we then look for a compatible LTB condition. While this is done in \cite{Bojowald:2008ja,Bojowald:2009ih} by using the effective Hamiltonian, here we perform our generalized analysis at the level of the equations of motion. The reason for choosing this approach is that the equations have a much simpler structure thus algebraic manipulations and numerical analysis are less complicated. Using this approach in contrast to \cite{Bojowald:2008ja,Bojowald:2009ih} we are able to determine compatible LTB conditions for models in which holonomy and inverse triad corrections are present simultaneously. Furthermore, our results show that some of the results derived in  \cite{Bojowald:2008ja,Bojowald:2009ih} need to be
improved in the way how the contribution of the spin connection is considered as well as whether the LTB condition can be used as a gauge fixing condition. By doing so we obtain partly different conclusions to \cite{Bojowald:2008ja,Bojowald:2009ih}. Another example where we also extend the results of \cite{Bojowald:2008ja,Bojowald:2009ih} is that the analysis used here allows polymerized models that have the classical LTB condition as a compatible one, whereas for the models in  \cite{Bojowald:2008ja,Bojowald:2009ih} the LTB condition of the effective model is always non-trivially differing from the classical one.  
~\\
~\\
Another aspect in our investigation is that we also consider 
constraints on the possible polymerizations that come from the requirement that the classical algebra of the spatial diffeomorphism constraint and the geometric contribution of the Hamiltonian constraint agrees with its polymerized counterpart.
Note that this requirement is not the same as asking for the algebra of the total constraints to close. The latter is usually understood as ensuring the covariance in the Hamiltonian formalism. In classical models with dust we can consider a temporal gauge fixing with respect to dust time, see for instance \cite{Giesel:2020raf,Husain:2011tm} for models in quantum cosmology and \cite{Kuchar:1991pq,Kuchar:1995xn} for models in full GR or LQG \cite{Giesel:2007wn} where this has be done. Working in such a partially gauge fixed models at the classical level the fact that the total Hamiltonian constraints are first class carries over to the property that the geometric contribution to the Hamiltonian constraint Poisson commute up to the spatial diffeomorphism. This means that the spatial diffeomorphism constraint and the geometric contribution of the Hamiltonian constraint form a closed algebra. If we demand this property also to hold for the effective model, then it turns out to be exactly the above mentioned assumption for the algebra. The question of conditions on the algebra can also be discussed in the context of mimetic gravity models.  These are covariant models which involve next to the
metric, a scalar field and a Lagrange multiplier that enforces the mimetic condition, together
with a mimetic potential, that is a term in the action that involves higher derivative couplings between gravity and the scalar field. As shown in \cite{Langlois:2017hdf,Bodendorfer:2017bjt,BenAchour:2017ivq,deHaro:2018hiq,Bodendorfer:2018ptp,Han:2022rsx},  a choice of polymerization in a given effective model can be associated with the choice of a particular mimetic model, even if this correspondence is not unique. 
If we choose a temporal gauge fixing with respect to the involved scalar constraint, then there exist models for which the Hamiltonian constraint deparametrizes and for which it does not. As we can link the choice of polymerization to a a choice of the mimetic potential we can formulate the condition on the algebra in the context of mimetic models as follows: The  condition on the algebra that is satisfied in the classical theory selects only those polymerization that correspond to mimetic models for which the Hamiltonian constraint deparametrizes. If we drop that assumption we also allow polymerizations corresponding to mimetic models where no deparametrization occurs. Nevertheless these models are intrinsically consistent covariant models, so we expect covariance to be implemented also in the Hamiltonian formulation. What we relax is that once we choose a temporal gauge fixing the quantities that become the physical Hamiltonian densities once also the diffeomorphism constraint is reduced do not commute in the effective model, although they do in the classical theory.
~\\
~\\
On the basis of the results in this article we have applied the developed formalism to two further cases which are discussed in two separated articles \cite{GenBirkhoff} and \cite{noShocksPaper}. The first one deals with the vacuum case for which we can formulate a Birkhoff-like theorem in the context of effective models. The work in \cite{noShocksPaper} applies our results to the dust collapse in effective models with a particular focus on the questions under which assumptions shock solutions exists and what we can learn about the properties of the effective model if we consider different choices of coordinates.
~\\
~\\
The paper is structured as follows: In table \ref{tab:Notation} we provide a list of the notation used in this article that will be particularly relevant for section \ref{sec:EffMod}. 
After the introduction in section \ref{sec:intro}, we discuss the embedding of LTB models into spherically symmetric models in the classical theory in section \ref{sec:ClassMod}, where we discuss the marginally and non-marginally bound models separately in subsections \ref{sec:ClassNonMarg} and \ref{sec:ClassMarg}. Because we extend the analysis on the LTB condition in \cite{Bojowald:2008ja,Bojowald:2009ih}, we discuss in section \ref{sec:CompBojowald} in detail how our formalism compares and differs from the strategy followed in \cite{Bojowald:2008ja,Bojowald:2009ih} and in particular what our analysis generalizes. Guided by the results in the classical theory, section \ref{sec:EffMod} presents the derivation of compatible LTB conditions in effective models. We start in section \ref{sec:GenRemarks} with a discussion of what kind of properties we want the effective models under consideration to satisfy and in addition give a brief summary of the main results of section \ref{sec:EffMod} for the benefit of the reader before we discuss the detailed analysis. In section \ref{sec:HolCorrections} we investigate what kind of restrictions we obtain on the polymerization if we require that the algebra of the effective geometric contribution of the Hamiltonian constraint and the spatial diffeomorphism constraint form a closed algebra. This selects a certain class of effective models covered by Lemma \ref{lemma:f1_CC}. The derivation of compatible LTB conditions is presented in section \ref{sec:LTBcondition}. Here we do not only consider effective models covered by Lemma \ref{lemma:f1_CC} but also effective models with more generic polymerizations. As two applications of our formalism we discuss in section \ref{sec:examples} the models considered in \cite{Bojowald:2008ja,Bojowald:2009ih} in subsection \ref{sec:example_1} and the model from \cite{Tibrewala:2012xb} in subsection \ref{sec:example_2}. The analysis in subsection \ref{sec:example_1} allows to discuss in detail where the similarities and differences are compared to the results in \cite{Bojowald:2008ja,Bojowald:2009ih} are. The model in subsection \ref{sec:example_2} is an example of the models that fall into the class of models of Lemma \ref{lemma:f1_CC} and is also covered by Corollary \ref{cor:C_and_dynamics}. Finally we summarize and conclude in section \ref{sec:Concl}.
\section{LTB conditions in  spherically symmetric models: Classical case}
\label{sec:ClassMod}
In this section we briefly review the restriction from classical spherically symmetric models to LTB models. For this purpose a so-called LTB condition is introduced which restricts the spherically symmetric sector to the LTB one. Early works on LTB conditions in the context of loop quantum gravity can be found in \cite{Bojowald:2008ja,Bojowald:2009ih}, where two LTB conditions were introduced. 
We will discuss the marginally and non-marginally bound case separately because the way LTB conditions can be implemented differs in the two cases and to the knowledge of the authors this has not been considered in detail in the literature yet. Finally, we comment on the differences to and some open issues in the results in \cite{Bojowald:2008ja,Bojowald:2009ih}.

\noindent We work with the Ashtekar-Barbero variables $(A_a^j, E^a_j)$  in spherical symmetry and choose  the spatial manifold as $\mathbb{R}\times S^2$. We consider the Gaussian dust model \cite{Kuchar:1990vy} in spherical symmetry and denote the degrees of freedom by $(T,P_T, S,P_S)$. In this symmetry reduced case after implementing the Gau\ss{} constraint they have the following form
\begin{eqnarray*}
A_a^j \tau_j \mathrm{~d} X^a & = & 2\beta K_x(x) \tau_1 \mathrm{~d} x+\left(\beta K_\phi(x) \tau_2+\frac{\partial_x E^x(x)}{2E^\phi(x)} \tau_3\right) \mathrm{d} \theta \\
&&+\left(\beta K_\phi(x) \tau_3-\frac{\partial_x E^x(x)}{2E^\phi(x)} \tau_2\right) \sin (\theta) \mathrm{d} \phi+\cos (\theta) \tau_1 \mathrm{~d} \phi \\
E_a^j \tau^j \frac{\partial}{\partial X^a} & = &E^x(x) \sin (\theta) \tau_1 \partial_x+\left(E^\phi(x) \tau_2\right) \sin (\theta) \partial_\theta+\left(E^\phi(x) \tau_3\right) \partial_{\phi}\,,  
\end{eqnarray*}
where $X^a=(x,\theta,\phi)$, $a=1,2,3$ denote spherical coordinates, $\beta$ the Barbero-Immirzi parameter and $\tau_j=-\frac{1}{2}\sigma_j$ with $\sigma_j$ being the Pauli matrices. The phase space which is already partially gauge fixed with respect to the Gau\ss{} constraint includes $\left(K_x(x), E^x(x)\right)$ and $\left(K_{\phi}(x), E^{\phi}(x)\right)$ and the dust sector and their non-vanishing Poisson brackets read
\begin{eqnarray*}
 \{K_x(x), E^x(y)\} &=& G \delta(x,y)\quad  \{K_\phi(x), E^\phi(y)\} =G \delta(x,y),\\
 \{T(x),P_T(y)\}&=&\delta(x,y), \quad \{S(x),P_S(y)\}=\delta(x,y)
\end{eqnarray*}
with $G$ being Newton's constant. The metric in terms of these variables reads
\be 
\label{eq:metricSphSymm}
\mathrm{d} s^2 = -N(x,t)^2 dt^2 + \frac{(E^{\phi})^2}{\abs{E^x}} (dx + N^x dt)^2 + E^x d \Omega^2\,.
\ee 
The total classical Hamiltonian constraint $ C^{\rm tot}$ is given by
\begin{eqnarray*}
 C^{\rm tot}(x) &=& C(x)+ C^{\rm dust}(x)
 \end{eqnarray*}
 with
 \begin{eqnarray}
 \label{eq:Cdust_C}
     C^{\rm dust}(x)&=&\qty(P_T \sqrt{1+\frac{\left|E^x\right|}{\left(E^{\phi}\right)^2}\left(T^{\prime}\right)^2}+\left(\frac{\left|E^x\right|}{\left(E^{\phi}\right)^2} T^{\prime} P_S S^{\prime}\right)\left(1+\frac{\left|E^x\right|}{\left(E^{\phi}\right)^2}\left(T^{\prime}\right)^2\right)^{-\frac{1}{2}})(x)
\end{eqnarray}
and the  geometric contribution to the Hamiltonian constraint $C$ has the following form
\be\label{eq:C_c}
C(x)=\frac{1}{2 G}\frac{E^\phi}{\sqrt{{{E^x}}}}\left[-E^{x}\qty(  \frac{4 K_x K_{\phi}}{E^{\phi}}+ \frac{K_{\phi}^2}{E^x}  ) + \qty(\frac{  {{E^x}}'}{2{{E^{\phi}}} })^2 - 1 + 2\frac{E^x}{E^\phi}\qty(\frac{   {{E^x}}'}{2{{E^{\phi}}} })'
\right](x)\,,
\ee
where the prime stands for a derivative with respect to the radial coordinate $x$. Note that we expressed the involved spin connection $\Gamma_\phi$ directly in terms of the triads and its derivatives as $\Gamma_{\phi}=-\frac{\left(E^x\right)^{\prime}}{2 E^{\phi}}$. The total classical spatial diffeomorphism constraint $C^{\rm tot}_x$ reads
\begin{eqnarray*}
C^{\rm tot}_x(x) &=& C_x(x)+ C^{\rm dust}_x(x),
\end{eqnarray*}
with the individual geometric and dust contributions  $C_x$ and $C^{\rm dust}_x$ respectively given by
\be\label{eq:diffeo}
C_x(x)= \frac{1}{G}\left(E^{\phi} {K_{\phi}} ' - K_x {E^x}'\right)(x)\quad{\rm and}\quad C^{\rm dust}_x=\left(P_T T^\prime + P_S S^\prime\right)(x).
\ee

\noindent LTB models are spherically symmetric solutions of Einstein's equation with pressureless dust. The LTB metric is in diagonal form in terms of our chosen variables 
\begin{equation}
\label{eq:MetricLTBDiag}
 \mathrm{d} s^2=-\mathrm{d} t^2+\frac{\left((E^x)^\prime\right)^2}{4|E^x|\Xi^2(x)} \mathrm{d} x^2+|E^x| \mathrm{~d} \Omega^2\,,  
\end{equation}
where we introduced $\Xi(x):=\sqrt{1+{\cal E}(x)}$ with ${\cal E}(x)$ being the usual function involved in the LTB metric that vanishes in the marginally bound case, see for instance \cite{Bojowald:2008ja} where this notation was used. In general also ${\cal E}(x) \equiv \mathrm{const}$ can be chosen in the marginally bound case, since then we can rescale the radial coordinate to set $\widetilde{\Xi}(x) \equiv 1$. Comparing the metrics in \eqref{eq:metricSphSymm} and \eqref{eq:MetricLTBDiag} and following \cite{Bojowald:2008ja} the so-called LTB condition  for the triads denoted as $C_{\rm LTB}$ in the classical theory is given by
\begin{equation}
\label{eq:LTBCondTriads}
 C_{\rm LTB}(x) = \left(|E^x|^\prime -2\Xi(x)E^\phi\right)(x)\, . 
\end{equation}
In the following we will discuss how $C_{\rm LTB}$ can be implemented in the marginally and non-marginally bound case which then in both cases allow to restrict the classical spherically symmetric model to its LTB sector. The analysis in the classical model serves as preparation for the case of effective models discussed in section \ref{sec:EffMod}.

\subsection{The non-marginally bound case}
\label{sec:ClassNonMarg}
As we saw in the previous section by comparing the form of the LTB metric in  \eqref{eq:MetricLTBDiag} with the general spherical symmetric one in \eqref{eq:metricSphSymm}, the LTB condition for the triad in \eqref{eq:LTBCondTriads} is an additional condition to the spherically symmetric model that implements a specific relation between the two triad variables that are independent in the spherically symmetric model. In the non-marginally bound case we can understand the LTB condition in \eqref{eq:LTBCondTriads} as a gauge fixing condition for the spatial diffeomorphism constraint as they form a second class pair. We can further see from the desired form of the lapse function and the shift vector in the LTB model, that is $N = 1$, $N^x=0$, that choosing a dust time gauge is convenient. Therefore,  we will consider the dust time gauge together with the LTB condition from \eqref{eq:LTBCondTriads} as gauge fixing conditions in the spherically symmetric model denoted by $G_T$ and $G_x$ respectively
\begin{eqnarray}
    \label{eq:GFconditionsclassical}
  G_T = (T -t)(x) && \quad G_x = \frac{{E^x}'}{2E^\phi}(x) - \ltbf(x)\,.     
\end{eqnarray}
Note that in the non-marginally bound LTB model $\ltbf' \neq 0$. Further we restrict this analysis to positive radial triad components, allowing us to drop the absolute value,  present in \eqref{eq:LTBCondTriads}, in $G_x$. The stability equations for these gauge fixing conditions are
\begin{align}
    \dv{G_T(x)}{t} &= \poissonbracket{G_T(x)}{H_\mathrm{can}} -1 = \left(N \sqrt{1+\frac{E^x}{\left(E^{\phi}\right)^2}\left(T^{\prime}\right)^2} + N^x\, T'\right)(x) \approx N(x)-1 \overset{!}{=}0\\
    \dv{G_x(x)}{t} &= \poissonbracket{G_x(x)}{H_\mathrm{can}} \approx (\zeta[N] + N^x\, \ltbf')(x)\approx (N^x\, \ltbf')(x)\overset{!}{=}0.
\end{align}
 Note that the functional $\zeta$ appearing in the second stability equation is for constant lapse functions weakly  satisfying
\begin{equation}\label{eq:zetanonmarginallybound}
    \zeta[N](x) \approx  \left(N\frac{\sqrt{E^x}G}{{E^\phi}^2 } C_x\right)(x) \approx -  \left( N \frac{\sqrt{E^x}G}{{E^\phi}^2 }C^{\rm dust}_x\right)(x) \approx - \left(N\frac{\sqrt{E^x}G}{{E^\phi}^2 }\qty(P_S S')\right)(x)\,,
\end{equation}
where we used the temporal gauge $T'\approx t'=0$ in the last weak equivalence. For the full result see equation \ref{eq:appendixpoissonGxC} in appendix \ref{app:ltb_general}. This means we have to constrain either $P_S = 0$ or $S'=0$ in order to get the solution $N^x = 0$ for the shift vector. In both cases this just adds another first class constraint to the model that forces the field $S$ to become constant under the Hamiltonian equations of motion. This holds true for the Brown-Kucha\v{r} \cite{Brown:1994py} dust model in an analogous way. In general we have to constrain the matter contribution to the diffeomorphism constraint without the time reference field (since this is already weakly vanishing due to the time gauge) to vanish. In case of dust models which incorporate only one matter field, like non-rotational dust, this additional first class constraint is absent. Hence, the final number of physical degrees of freedom agrees for dust models with one and two dust fields once the LTB condition and the temporal gauge fixing have been implemented.
\newline
We can also see that the non-marginally bound condition, i.e. $\ltbf' \neq 0$, is crucial as otherwise $G_x$ cannot be treated as a gauge fixing condition for the diffeomorphism constraint. Next we want to investigate the dynamical structure of the reduced system by computing the Dirac brackets with respect to the  system of second class constraints $C^J := (G_T, G_x, C^\mathrm{tot}, C_x^\mathrm{tot})^{T}$. The non-vanishing Dirac brackets of the elementary gravitational phase space variables are given by
\begin{align}
\label{eq:DBGravPart}
\poissonbracket{K_x(x)}{E^x(y)}_D &= G \qty(1 + \frac{\ltbf}{\ltbf'}(y)\pdv{y})\delta(x,y), && \poissonbracket{K_\phi(x)}{E^\phi(y)}_D = G \qty(1 + \frac{\ltbf}{\ltbf'}(x)\pdv{x})\delta(x,y)\\ \poissonbracket{K_x(x)}{K_\phi(y)}_D &= -G^2 \frac{C_x^{\mathrm{dust}}}{2 {E^\phi}^2 \ltbf'}\pdv{y}\delta(x,y), && \poissonbracket{E^x(x)}{E^\phi(y)}_D = 0 \\ \poissonbracket{K_x(x)}{E^\phi(y)}_D &= \frac{G}{2}\pdv{x}\qty( \frac{1}{\ltbf'(x)}\pdv{x}\delta(x,y)), && \poissonbracket{K_\phi(x)}{E^x(y)}_D = -2 G \,\frac{\ltbf^2}{\ltbf'}(x)\delta(x,y)\,.
\end{align}
More details of the computations are presented in appendix \ref{appendix:Diracbrackets_classical}. As we can see, the results of these Poisson brackets partly involve phase space variables as well as the LTB function and  derivatives acting on the delta function. Note that in case of only one (temporal) reference field or in the case that we look at a system where the matter part (without the temporal reference field) of the diffeomorphism constraint is already implemented as a primary constraint, the bracket $\poissonbracket{K_x(x)}{K_\phi(y)}$ vanishes weakly. The set of variables $(K_\phi,E^x)$ is somehow special as it is the only pair of variables for which the result of the Dirac bracket involves no derivatives acting on the delta function. Therefore, for this subset of variables we can absorb the terms involving the LTB function with a suitable canonical transformation into a redefinition of the elementary canonical variables which then satisfy standard canonical Poisson brackets. The latter is important for a later quantization of the model. In turn, for any other chosen set of canonical variables a quantization of the gauge fixed model is more complicated because we need to find a representation that respects the algebra of the Dirac brackets.
Since the result of this Poisson bracket involves the LTB function here, we discuss in \eqref{eq:Diracbracket_simplified}  how this result generalizes when we consider it in the context of effective models.
~\\
~\\
To conclude this section we want to derive the corresponding LTB Hamiltonian for the chosen gauge fixings. For the final choice of independent elementary variables there
are multiple reductions possible, since the constraints can be used to eliminate appearing phase space variable for any given choice at the classical level. As we aim at analyzing the corresponding effective model later on, taking the aforementioned algebra of Dirac brackets into account, choosing the canonical pair $\qty(K_\phi, E^x)$ is especially convenient as far as the corresponding quantum model is concerned. From the LTB condition and diffeomorphism constraint in the case that $P_S S' \equiv 0$, we can deduce
\begin{align}\label{eq:EandKLTBreductionnonmarginal}
    E^\phi(x) \approx\frac{(E^x)'}{2\ltbf}(x), && K_x \approx \frac{(K_\phi)'}{2\ltbf}(x)\,.
\end{align}
Using these relations in addition to the temporal gauge fixing condition and solution for the Lagrange multipliers we can write the corresponding, partially gauge fixed, canonical Hamiltonian as
\begin{equation}
\label{eq:PGFHcanClass}
    H_\mathrm{can} \approx \int \mathrm{d}x\, C^\mathrm{tot}(x) \approx \int\mathrm{d}x\, \qty(\frac{1}{2G}\frac{\ltbf'}{\ltbf^2}\sqrt{E^x}\qty(\ltbf^2 - K_\phi^2 - 1)  + P_T)(x)\,.
\end{equation}
 The equations of motion of this physical Hamiltonian are given by
 \begin{align}
     \Dot{{E^x}}(x) &= \dv{E^x}{t}\,(x) =  -2 K_{\phi}\sqrt{E^x}(x) \\
     \Dot{{K_\phi}}(x) &= \frac{1}{2\sqrt{E^x}}\qty(K_\phi^2 + 1 - \ltbf^2)(x)\,.
 \end{align}
\subsection{The marginally bound case}
\label{sec:ClassMarg}
The marginally bound case is defined by $\ltbf' = 0$ from which we will set without loss of generality $\ltbf(x)\equiv1$. This has an impact on the constraint algebra since now we obtain for $G_x$ in \eqref{eq:GFconditionsclassical} using equation \eqref{eq:appendixpoissonGxdiffeo}
\begin{eqnarray*}
 \poissonbracket{G_x(x)}{C_x(y)}\approx 0 .
\end{eqnarray*}
Consequently, the LTB condition  $C_{\rm LTB}$ which is weakly equivalent to $G_x$ can not be used as gauge fixing condition for the the spatial diffeomorphism constraint $C_x$ any longer. However, we can consider a spherically symmetric model in which $C_{\rm LTB}$ is treated as an additional constraint in the system by introducing an additional Lagrange multiplier that enforces this constraint. In this case the constraint analysis results in only non-vanishing Poisson bracket between constraints as: 
\begin{equation}
    \poissonbracket{G_x(x)}{C^\mathrm{tot}\qty[N]} \approx \qty(N \frac{\sqrt{E^x}G}{{E^\phi}^2 } C_x + N'K_\phi \frac{\sqrt{E^x}}{ E^\phi} )(x).
\end{equation}
Transferring to the partially gauge fixed system where the dust time gauge is implemented, we can see that only the first term survives. If we consider for instance non-rotational dust, then only the dust field $T$ is present and in the time gauge chosen $C_x$ weakly vanishes and thus $G_x$ is a first class constraint. This is an issue since in comparison with the non-marginally bound case here the diffeomorphism constraint and LTB condition are two first class constraints instead of being a pair of second class constraints and thus no physical degrees of freedom are left after reduction of constraints. Note that this problem can not be circumvented by considering Gaussian or Brown-Kucha\v{r} dust, where, as can be seen in \eqref{eq:Cdust_C} and \eqref{eq:diffeo}, another dust field $S$ is present. The reason is that in this case a new constraint arises in the constraint analysis, namely $C_x \approx -C_x^\mathrm{dust} \approx -P_S S'\approx 0$. It turns out that this constraint is first class as well, and therefore it reduces the additional dust degrees of freedom, which means which means that finally there are no physical degrees of freedom in these systems either. As can be seen from \eqref{eq:PGFHcanClass} the geometric part of the partially gauge fixed canonical Hamiltonian vanishes in the marginally bound case. In the following we discuss in more detail the close relation between the marginally bound case and the stationary and vacuum solution in classical theory as well as in the framework of effective models in \cite{GenBirkhoff}.
~\\
~\\
\subsection{Discussion on and comparison to the results of Bojowald et al. in \texorpdfstring{\cite{Bojowald:2008ja,Bojowald:2009ih}}{}}
\label{sec:CompBojowald}
Our investigation from the last two subsections show that the constraint analysis and the implementation of an LTB condition in the classical theory as well as in the polymerized effective model is more subtle than discussed in the work in \cite{Bojowald:2008ja,Bojowald:2009ih}. 
Here we will list a few aspects that underlie this using the results obtained so far and some points that address open issues if we go beyond the classical theory and work with effective models.
\begin{itemize}
\item{\underline{Constraint analysis and LTB condition:}}
~\\
The first observation in the classical theory is the way one must deal with the LTB condition in \eqref{eq:LTBCondTriads} is different for the marginally and non-marginally bound case. In the latter case it builds a pair of second class constraints  with the diffeomorphism constraint, whereas in the first case this is not given.  Consequently, for the marginally bound case, when in addition one works in dust comoving coordinates, one can only add this LTB condition as an additional first class constraint to the spherically symmetric system. However, as our constraint analysis shows in this case this yields a system that possesses no physical degrees of freedom. To the understanding of the authors this has not been considered in \cite{Bojowald:2008ja} where for the marginally bound case the LTB conditions is treated as a second class constraint.
\item{\underline{One versus two LTB conditions:}}
~\\ 
In case we use the LTB condition  in \eqref{eq:LTBCondTriads} in the non-marginally bound case as a gauge fixing condition for the diffeomorphism constraint and also consider comoving dust coordinates, then as discussed above, we have two gauge fixing conditions -- one for the Hamiltonian and one for the diffeomorphism constraint. In \cite{{Bojowald:2008ja,Bojowald:2009ih}} it is argued that from the stability of the LTB condition  in \eqref{eq:LTBCondTriads} follows a second LTB condition for the extrinsic curvature variables which is identical to the second equation in \eqref{eq:EandKLTBreductionnonmarginal}.  Now if we consider the two gauge fixing conditions, here shown in \eqref{eq:GFconditionsclassical}, then the stability of both is ensured by choosing the Lagrange multipliers $N,N^x$ appropriately. As our analysis shows, the stability of both relate the shift vector $N^x$ to the dust contribution of the diffeomorphism constraint $C_x^{\rm dust}$. Now in case we have one dust field in the model and use comoving dust coordinates we have $C_x^{\rm dust}=PT'\approx 0 \approx  -(K^\prime_\phi -2\ltbf K_x)$ without any further conditions, like a second LTB condition, to require. In other words, implementing the pair of second class constraints encoded in  LTB condition in \eqref{eq:LTBCondTriads} and diffeomorphism constraint to remove two canonical variables  is equivalent to implementing the two LTB conditions considered in \cite{Bojowald:2008ja,Bojowald:2009ih}. These are the one for triads shown in \eqref{eq:LTBCondTriads} and the second equation for the extrinsic curvature shown in \eqref{eq:EandKLTBreductionnonmarginal}, where one has to work with the Dirac brackets in \eqref{eq:DBGravPart} after the implementation of these constraints.
 
This is completely different from treating the second LTB condition as an additional constraint/gauge fixing condition.  It seems that in \cite{Bojowald:2008ja,Bojowald:2009ih} a model with only one dust field is considered because there is no contribution from the dust in the diffeomorphism constraint. Thus, to the understanding of the authors the second LTB condition introduced in \cite{Bojowald:2008ja,Bojowald:2009ih} is automatically satisfied.
\item{\underline{Spin connection in stability analysis of the LTB condition:}}
~\\
As shown by our analysis in the marginally bound model with a comoving dust frame the LTB condition for the triads in \eqref{eq:LTBCondTriads} is a first class constraint in the partially gauge fixed theory where the lapse $N=1$. Hence, in the stability analysis of the LTB condition as performed  in \cite{Bojowald:2008ja} the spin connection can only be set to $\Gamma_\phi=-1$, the value it takes in the LTB sector, right from the beginning as it seems to be done in \cite{Bojowald:2008ja}, if one works with $N=1$ or at least a constant lapse function in the classical case. We see in the next section that taking this point seriously into account is crucial for the analogue stability analysis in the effective models. In the non-marginally bound case the LTB condition and the diffeomorphism constraint build a second class constraint pair. These constraints can be inserted into the Hamiltonian constraint provided that we use the corresponding Dirac bracket instead of the Poisson bracket.  As \eqref{eq:DBGravPart} shows the algebra of the Dirac brackets is more complicated and differs from the standard canonical commutation relations (CCR) the corresponding Poisson brackets satisfy. It is not obvious to the authors that this has been taken into account in \cite{Bojowald:2009ih} when they discuss that their defined reduced constraint preserves the LTB condition(s).

\item{\underline{Choice of variables in the LTB sector:}}
~\\
Related to the last point, the algebra of the Dirac brackets in the gravitational sector shown in \eqref{eq:DBGravPart} indicates that there is a preferred choice of phase space coordinates when we restrict the spherically symmetric sector to the LTB sector, namely to choose $(K_\phi,E^x)$ as the elementary canonical pair in the LTB sector. The reason for this  choice is that it is the only pair in the gravitational sector that satisfies, after a suitable canonical transformation, standard CCR for which the usual quantization procedure can be applied. For other combinations one first needs to show that corresponding representation for the corresponding operators exits. This becomes particularly relevant in the case of the effective models which should mimic the underlying quantum theory to some extend and and are usually based on the assumption that one works with a quantization of the standard (symmetry-reduced) holonomy-flux algebra and not some more complicated algebra.
\end{itemize}
As in \cite{Bojowald:2008ja,Bojowald:2009ih} the main purpose of our work is to consider an LTB condition beyond the classical theory which is consistent with the effective dynamics and we will discuss such an analysis in detail in section \ref{sec:EffMod}.  Some of the motivation how we perform such investigation comes from the following open questions we want to discuss in the context of the  work in \cite{Bojowald:2008ja,Bojowald:2009ih}:

\begin{itemize}
    \item{\underline{Effective LTB conditions:}}
~\\ 
According to the results in \cite{Bojowald:2008ja,Bojowald:2009ih} the classical LTB conditions they introduce are no longer stable with respect to their defined effective dynamics which is obtained by either introducing holonomy corrections or inverse triad correction in the dynamics. Their strategy to deal with this fact is that they modify the LTB condition for the triads in \eqref{eq:LTBCondTriads} as well as the equation for the extrinsic curvature in \eqref{eq:EandKLTBreductionnonmarginal}, that they call second LTB condition in their work and those modifications are restricted by the stability requirement of the LTB conditions. As we have seen that the way how the constraints are treated differs for the non-marginally and marginally bound case. We aim at carefully analyzing the stability of LTB condition(s) in effective models taking this insight into account and understand if and how the results differ from the existing results in \cite{Bojowald:2008ja,Bojowald:2009ih}.
\item{\underline{The role of the spin connection in the stability analysis:}}
~\\
The role of the spin connections becomes  more involved when the effective models are considered, which are investigated in section 4 in \cite{Bojowald:2008ja} and section 3 in \cite{Bojowald:2009ih}. In the case of the marginally bound model in \cite{Bojowald:2008ja} the  spin connection in the dynamics is set to $\Gamma_\phi=-1$ right from the start of the analysis. However, then the LTB condition in \eqref{eq:LTBCondTriads}  is modified such that even in the effective LTB sector this is no longer valid but we have  either $\Gamma_\phi \approx -f(E^x)$ or $\Gamma_\phi \approx g(K_\phi)$, depending on the chosen modification in \cite{Bojowald:2008ja}. However, the generator of the dynamics is still defined with a spin connection $\Gamma_\phi=-1$. Although it is discussed in section 4.2 of \cite{Bojowald:2008ja} that those corrections need to be included for consistency but in general not so easy to compute a strategy how to deal with this is not given in \cite{Bojowald:2008ja}. Hence, whether a simpler form of the LTB condition in effective models exist is an interesting question from this point of view.
Likewise in the non-marginally bound case in \cite{Bojowald:2009ih} in section 3.1 and 3.2 the spin connection is also fixed to its classical LTB value at the start of the stability analysis. In section 3.3 a modified spin connection is considered as well as the requirement that the resulting Hamiltonian constraints should commute. Here, one would need to consider the corresponding Dirac brackets associated with these modified LTB conditions in the stability analysis if we insert the second class constraints right from the beginning. 

\item{\underline{The role of an effective LTB metric:}}
~\\
With the insights discussed above we wish to analyze in particular the question whether a polymerization of the dynamics in the effective models necessarily leads to an effective LTB metric which merges only in the classical limit into the classical LTB metric. The results in \cite{Bojowald:2008ja,Bojowald:2009ih} indicate this because in their effective models the stable version of the LTB condition in \eqref{eq:LTBCondTriads} 
necessarily involves a modification which carries over to the form of the corresponding metric. As our analysis in the next section will show, there exists also effective models with a polymerized dynamics that are still consistent with the classical LTB metric.
\end{itemize}
 
\section{LTB conditions in  spherically symmetric models: Effective model case}
\label{sec:EffMod}
The discussion in the former section showed that the way how the LTB condition is implemented is crucial in order to perform a consistent reductions from the spherically symmetric sector to the LTB sector. If one now wishes to perform a similar analysis in effective models, further complications arise, such as the fact that generic effective constraints may not satisfy the analog of the classical constraint algebra and how to deal with the freedom to choose a polymerization for the constraint and the LTB condition. Therefore, we want to do the analysis of the effective models systematically, so that we do not just study a particular effective model, but come up with a list of properties that the effective models must satisfy, and then discuss examples for models which fall into this class. 

\subsection{Generic remarks for the effective models under consideration}
\label{sec:GenRemarks}
Before we  discuss specific derivations and applications in the next section we wish to start with a few remarks on how we will perform the stability analysis of the LTB condition and what class of effective models we will consider.
\begin{itemize}
    \item {\underline{Function preserving polymerization:}}
    ~\\
    We choose this requirement to limit the huge freedom one has by introducing a polymerization function and to stay as close as possible to the usual quantization procedure.
\item{\underline{Closure of the algebra of effective constraints:}}
    ~\\
    In accordance with GR the starting point of all effective models we want to consider should be a general covariant theory. In the Hamiltonian formulation of GR this covariance is encoded in the algebra of four first class constraints -- the hypersurface deformation algebra. They are the generators of gauge transformations which are time- and spacelike diffeomorphisms on the phase space. If we choose temporal partial gauge fixing in GR with respect to a dust field, then the dust models considered here have the property that the Hamiltonian constraint deparametrizes. As a consequence, the corresponding physical Hamiltonian densities commute up to the spatial diffeomorphism. If we later quantize such a model, we usually aim at defining the operators corresponding to the physical Hamiltonian density in such a way that they also commute up to the spatial diffeomorphism in the quantum theory. Often this is even taken as a guiding principle in order to restrict possible ambiguities in the quantization step. Now if we consider a given polymerization we would like to consider only those which still allow covariance to be implemented in the polymerized model.\footnote{In LQC, another way to explore covariance is via a covariant action \cite{Olmo:2008nf,Delhom:2023xxp}.  } Recently, this has been discussed in the context of (extended) mimetic gravity models \cite{Chamseddine:2013kea,Sebastiani:2016ras,Takahashi:2017pje,Langlois:2017mxy}. These are covariant models which involve next to the metric, a scalar field and a Lagrange multiplier that enforces the mimetic condition, together with a mimetic potential, that is a term in the action that involves higher derivative couplings between gravity and the scalar field. The (extended) mimetic gravity theory belongs to the family of Degenerate Higher-Order
Scalar-Tensor (DHOST) theories \cite{BenAchour:2016fzp,Langlois:2017mxy}, which propagates (up to) only three
degrees of freedom: one scalar and two gravity tensorial modes. In \cite{Langlois:2017hdf,Bodendorfer:2017bjt,BenAchour:2017ivq,deHaro:2018hiq,Bodendorfer:2018ptp,Han:2022rsx} it was shown that a specific choice of the mimetic potential can be related to a certain choice of the polymerization in effective models when a temporal gauge fixing with respect to the scalar field is considered. As all these models are derived from a covariant Lagrangian, due to the mimetic condition we expect that on the one hand the models involve second class constraints and on the other hand once we partially reduce with respect to these second class constraints we end up with a model involving a Hamiltonian and diffeomorphism constraint that are first class, similar to what happens in the Gaussian \cite{Kuchar:1991pq}  and Brown-Kucha\v{r} dust \cite{Kuchar:1995xn} models and in this sense the covariance is implemented. For the mimetic model originally introduced by Chamseddine and Mukhanov which involves a specific choice of the mimetic potential \cite{Chamseddine:2016uef}, this has been shown in \cite{Langlois:2015skt,BenAchour:2016fzp,Bodendorfer:2017bjt,deCesare:2020got}. For generic mimetic models, although one expects  similar properties, to the knowledge of the authors such a result is not available in the literature yet. When we choose a temporal gauge fixing for such models, the closure of the first class constraint algebra carries over to the condition that the associated Hamiltonian densities either commute up to the spatial diffeomorphism, if the system deparametrizes, or do not but in a controlled way, as we know what the non-vanishing result of their classical Poisson bracket is. If we consider polymerized effective models from the perspective of mimetic models, then for a generic mimetic potential the model will not deparametrize. Consequently, if we consider the corresponding polymerized physical Hamiltonian densities they will not commute up to the spatial diffeomorphism. Therefore, we can understand the requirement that the Hamiltonian densities commute up to the spatial diffeomorphism in this context as follows:  When we polymerize a classical model and require that in the temporally partial gauge fixed model the Hamiltonian densities commute up to the spatial diffeomorphism, then we restrict the possible classes of mimetic potentials and thus the possible choices of polymerization functions, because we only consider mimetic models in which the Hamiltonian constraint deparametrizes. If we drop that requirement we still have an underlying covariant mimetic model. Although the corresponding physical Hamiltonian density will not commute up to the spatial diffeomorphism, this corresponds to an intrinsically consistent mimetic model that however does not deparametrize. Following this perspective, and comparing classical and effective clocks as a simplified models for quantum clocks, we realize that deparametrization might not be preserved if the polymerization is not chosen carefully. In our analysis in section \ref{sec:EffMod} we will consider both cases, polymerization that correspond to a mimetic model that deparametrizes but also polymerizations where this is not the case.

\item{\underline{No polymerization of the spatial diffeomorphism constraint:}}
~\\
In most existing effective models, the spatial diffeomorphism constraint is treated differently from the Hamiltonian constraint as far as polymerization is considered, because unlike the Hamiltonian constraint, they do not consider polymerization for it.  In full LQG, the infinitesimal diffeomorphisms do not exist and one quantizes finite diffeomorphisms. An alternative approach is taken in the context of the master constraint or Brown-Kuchar dust model, in which the squared version of the infinitesimal diffeomorphism contracted with the spatial metric and weighted by a suitable density weight is quantized, which can be implemented as a well-defined operator in LQG. If we derive an appropriate effective model for the latter case, we also expect a polymerization effect in the diffeomorphism constraint in general. In this work, we will not consider such a possible polymerization effect in the diffeomorphism constraint. The reason for this is twofold. First, if we want to work with a first class algebra of effective constraints, the situation simplifies greatly because the classical diffeomorphism constraint has its usual geometric interpretation and transforms an effective polymerized Hamiltonian constraint in the expected way (given the fact that the polymerization did not change the density weight). Second and not completely unrelated to the first point, the geometric part of the diffeomorphism constraint becomes a conserved charge at the level of the physical Hamiltonian once the gauge fixings are implemented. As we will show in section \ref{sec:HolCorrections}, if we want to apply this to the effective model, no polymerization of the diffeomorphism constraint is allowed, since this leads to anomalies as far as the conserved charge is considered. We leave the more general class of models with polymerized diffeomorphism constraints to future work and consider the analysis here as a first step.

\item{\underline{Polymerization of the LTB condition:}}
~\\
We want to consider generic polymerizations of the LTB condition implementing the correct classical limit. In this way inverse triad corrections and holonomy corrections can be treated simultaneously. In \cite{Bojowald:2008ja} inverse triad and holonomy corrections were treated separately for the marginally bound case and modified LTB conditions were studied only for these separate cases. For the non-marginal bound case only inverse triad corrections were considered in \cite{Bojowald:2009ih}. Therefore, our rather general assumption about the polymerization of the LTB condition on the one hand allows us to go beyond the existing results in the literature and on the other hand provides the possibility to look at the existing results from a different angle. Since the second LTB condition used in \cite{Bojowald:2008ja,Bojowald:2009ih} is related to the (modified) LTB condition involving the triads, we do not need any assumptions for this relation, since all the chosen polymerization and inverse triad corrections for the effective Hamiltonian constraint as well as the polymerization of the LTB condition are transferred to the relation involving the extrinsic curvature through the effective dynamics.
\item{\underline{Stability of LTB condition:}}
~\\
We aim at performing the stability analysis at the effective level and therefore any modification of either the effective Hamiltonian or the LTB conditions due to polymerization or inverse triad effects needs to be consistently taking into account.
\end{itemize}
~\\
~\\
Before we  start the systematic analysis of effective models in the section \ref{sec:HolCorrections} and \ref{sec:LTBcondition},  we give a brief summary of what is included in these sections as well as the the main results presented there. For this we also refer to table \ref{tab:Notation} at the beginning of the article where the notation used in our work is listed. First we will analyze the implications of the second assumption, the closure of the algebra. Due to the fact that the diffeomorphism constraint remains the classical one in all effective models (third assumption), we want to consider in the following the only interesting bracket, which is the scalar constraint with itself. The other two brackets already close due to properties of the spatial diffeomorphism constraint. Requiring that the result of the remaining bracket is proportional to the diffeomorphism constraint restricts the form of compatible polymerizations of the scalar constraint. Most prominently, as we will show below,  we are not allowed to polymerize $K^x$, the radial component of the extrinsic curvature.
\\
As the next step we will conduct in section \ref{sec:LTBcondition} the analysis of the existence of a compatible LTB condition in the most general form with respect to the dynamics generated by an effective primary Hamiltonian in the form of LTB metric. As we will see this not only imposes conditions on the polymerization of the LTB condition but also the one of the polymerization in the effective primary Hamiltonian. The system is actually over-constrained for a given Hamiltonian, especially when there is a polymerization of $K^x$ involved. By this we mean that  for a given Hamiltonian -- and thus the resulting effective dynamics -- we might not be able to find a compatible polymerization of the LTB condition which can reduce the set of equations of motion to only two independent first order ones as expected for the LTB sector.
~\\
~\\
The main results of this work can be summarized as the following:
\begin{enumerate}
    \item[1.] The existence of a compatible LTB condition $g_{\Delta}(K_{\phi},E^x,\tilde{K}_x,\ltbf)(t,x) = \frac{\partial_x E^x}{E^{\phi}}(t,x)$ determines the form of polymerization of $K_x$ in the polymerization function $f$ completely. 
    \item[2.]
    In the non-marginally bound case a compatible LTB condition exists only when the involved polymerization functions do not contain a polymerization of $K_x$. 
\end{enumerate}    
The details of 1. and 2. are shown in Lemma \ref{thm:marginal} - \ref{thm:non_marginally bound case } and Corollary \ref{cor:gEx_m} - \ref{cor:gEx_non_m}.
\begin{enumerate}
    \item[3.] A compatible LTB condition ${g}_{\Delta} = \tilde{g}_{\Delta}(E^x) \ltbf$ in both the marginally and non-marginally bound case exists if the polymerization function $f$ is given by the form in Lemma \ref{lemma:f1_CC}. 
  \item[4.] Given a compatible LTB condition, 
  \begin{itemize}
  \item the dynamics is determined by the equations of motion of $K_{\phi},E^x$ in the LTB coordinates.
  \item the dynamics can be generated from a reduced effective Hamiltonian density $C^{\Delta} \sim C(K_{\phi},E^x)$ up to boundary terms in the non-marginally bound case with the Dirac bracket $\{ K_{\phi}, E^x \}_D = - \frac{2 \ltbf^2 \tilde{g}_{\Delta}}{\ltbf'}$, where the Dirac bracket is the one constructed in section \ref{sec:ClassNonMarg}.
  \item 
  the equations of motion for different radial coordinates $x$ decouple completely. 
\end{itemize}  
\end{enumerate}
The details of 4. are given in Corollary \ref{cor:C_and_dynamics}. 
\begin{enumerate}
    \item[5.] We generalize and improve the existing results in \cite{Bojowald:2008ja,Bojowald:2009ih}. First, the formalism introduced here for finding a compatible LTB condition allows to include holonomy and inverse triad corrections at the same time. Second, our results show that some of the results used in \cite{Bojowald:2008ja,Bojowald:2009ih} needs to be improved in the way the contribution of the spin connection are handled and if doing so we obtain partly different conclusions to \cite{Bojowald:2008ja,Bojowald:2009ih}. A further example where we extend the results is  that the analysis used here allows polymerized models that have the classical LTB condition as a compatible one. The detailed comparison of the results from \cite{Bojowald:2008ja,Bojowald:2009ih} and the ones obtained in this work can be found in section \ref{sec:example_1}.
\end{enumerate}
\subsection{Effective models with general holonomy polymerizations and inverse triad corrections}
\label{sec:HolCorrections}
To start the analysis of effective models we would like to include expected modifications to the classical gravitational contribution to the Hamiltonian constraint \eqref{eq:C_c} to capture the holonomy effects from LQG. The holonomy effects are encoded generally in polymerizations (via trigonometric functions), where linear combinations of $K_{\phi},K_{x}$ are polymerized with some function ${\cal P}$: $a K_{\phi} + b K_{x} \to {\cal P}(a K_{\phi} + b K_{x})$. The factors $a$ and $b$ here can be phase space dependent such that it can encode the possible $\bar{\mu}$-scheme, e.g. the one presented in \cite{Ashtekar:2006wn} (see \cite{Han:2022rsx} for its exposition in mimetic gravity). To capture such generic polymerizations to the $K_{\phi},K_{x}$ part appearing in the gravitational contribution to the  Hamiltonian constraint, we introduce the following ansatz 
\be\label{poly_f}
C \to C^{\Delta} : \;\;\frac{4 K_x K_{\phi}}{E^{\phi}} + \frac{K_{\phi}^2}{E^x} \to \left( \frac{4 K_x K_{\phi}}{E^{\phi}} + \frac{K_{\phi}^2}{E^x} \right) (1+f(K_x,K_{\phi},{E^x},{E^{\phi}} )),
\ee 
where the function $f$  can in principle include $K_{\phi},K_{x}$ in $C^\Delta$ to arbitrary combinations of polymerizations. The specific form is introduced such that one can easily recover the classical limit with a vanishing $f$, i.e. $f$ expanded as a polynomials in  the polymerization parameter $\alpha$ will not contain zeroth order terms in $\alpha$. For example, when we have a polymerization which brings $K_{\phi}$ to the form $\frac{\sin(\alpha K_{\phi})}{\alpha}$, the function $f$  is given by
\begin{eqnarray}
   f= \frac{ \frac{4 K_x \sin (\alpha K_{\phi})}{\alpha E^{\phi}} + \frac{\sin(\alpha K_{\phi})^2}{\alpha^2 E^x}}{ \frac{4 K_x K_{\phi}}{E^{\phi}} + \frac{K_{\phi}^2}{E^x} } -1 = \sum_{n = 1}^{\infty} (-1)^n \alpha^{2n} \frac{2 K_{\phi}^{2n} \left( 4^n K_{\phi}E^{\phi} + 4(n+1) K_x E^x  \right)}{ 3 (2n+2)! \left( K_{\phi}E^{\phi} + 4 K_x E^x \right)} .
\end{eqnarray}
Some other examples for the choice of $f$ will be given later in Lemma \ref{lemma:f1_CC} and section \ref{sec:examples}. 
We further introduce the inverse triad corrections to the rest of the Hamiltonian, that is those terms involving no extrinsic curvature, as follows:
\begin{equation}
\label{eq:InvTriadCorr}\frac{1}{\sqrt{E^x}} \to \frac{h_1(E^x)}{\sqrt{E^x}},\quad \frac{E^x}{\sqrt{E^x}}=\sqrt{E^x} \to \sqrt{E^x} h_2(E^x) .
\end{equation}
This results in the following polymerized gravitational contribution to the Hamiltonian constraint
\begin{equation}\label{eq:defpolyhamiltonianconstraint}
    C^{\Delta}(x)= \frac{ E^{\phi}}{2 G \sqrt{{{E^x}}}}\left[ -{(1 + f) E^x} \qty(\frac{4 K_x K_{\phi}}{E^{\phi}} + \frac{K_{\phi}^2}{E^x } ) + h_1 
\qty(\qty(\frac{  {{E^x}}'}{2{{E^{\phi}}} })^2 - 1   )+2\frac{E^x}{E^\phi} h_2 \qty(\frac{   {{E^x}}'}{2{{E^{\phi}}} })'\right](x)\,,
\end{equation}
where all holonomy corrections are encoded in $f$ and the inverse triad corrections of the terms not involving extrinsic curvature variables are encoded in $h_1,h_2$. Note that the function $f(K_x,K_{\phi},{E^x},{E^{\phi}} )$ can also encode inverse triad corrections as well, see for example the model in section \ref{sec:example_1}. This means if there are inverse triad corrections present in terms that include extrinsic curvature variables these will be involved in the polymerizations function $f$ next to possible holonomy corrections, whereas all remaining included inverse triad corrections will be covered by $h_1,h_2$ and by construction these terms do not contain any holonomy corrections. Nevertheless our ansatz will not cover all possible polymerized models, e.g. the model presented in \cite{Alonso-Bardaji:2021yls}. In the classical limit $f$ vanishes, whereas $h_1$ and $h_2$ are equal to 1. Further, the inverse triad modifications are encoded only in $E^x$ terms which is consistent with introducing these modifications only for the compact part of the spatial manifold. Note that for the same reason we do not introduce any modifications in $E^\phi$. In the marginally bound case, the terms proportional to $E^\phi$ drop out from $C^\Delta$ and it is only $E^x$ which contributes. \\

\noindent Here we will work with an effective metric in generalized Gullstrand–Painlevé 
coordinates 
\be \label{eq:metricSphSymm1} 
\mathrm{d}s^2 = - \mathrm{d}t^2 + \frac{(E^{\phi})^2}{E^x} (\mathrm{d}x + N^x \mathrm{d}t)^2 + E^x \mathrm{d} \Omega^2 \, ,
\ee 
where we have already implemented the comoving gauge as a partial temporal gauge fixing. 
As a result, a coordinate transformation on $t$ at the level of the metric is in principle not allowed unless we have a covariant Lagrangian underlying the effective theory. However we still want to keep the physics invariant under infinitesimal coordinate transformations along the radial direction $x$ generated by 
\begin{eqnarray} 
x \to x + \xi(t,x) \, .
\end{eqnarray} 
This implies that we are asking the effective system to be invariant under gauge transformations generated by the classical diffeomorphism constraint $C_x$. As a consequence, after partially gauge fixing with respect to the comoving gauge we will work with the primary Hamiltonian $H_P^{\Delta}$ 
\begin{eqnarray} \label{eq:defprimarypolyHamiltonian}
    H_P^{\Delta}[N^x] = \int \mathrm{d}x \, (C^{\Delta}+ N^x C_x)(x)  \,.
\end{eqnarray}

\noindent In order to have $C_x$ as a gauge symmetry of the model, we need to require that it is stable under the evolution generated by the primary Hamiltonian above. Since the diffeomorphism constraint has its classical form we know $\poissonbracket{C_x[N]}{C_x[M]}= C_x[[N,M]]$ and so the stability requirement reduces to
\begin{eqnarray}
    \poissonbracket{H_P^{\Delta}[N^x=0]}{ C_x(y)}= 0 \, .
\end{eqnarray}
This is equivalent to requiring that the polymerized Hamiltonian density still has density weight one, as we then have
\begin{eqnarray}\label{eq:commute_cx}
    \poissonbracket{C^{\Delta}[N]}{C_x[N^x]} =  C^{\Delta}[N^x \qty(\partial_x N)]\, .
\end{eqnarray}
As a consequence the polymerization functions must have density weight zero \cite{Giesel:2023euq}. In the case of the inverse triad corrections $h_1$ this is already true since it only depends on $E^x$, but for the holonomy corrections $f$ we have to restrict
\begin{eqnarray}\label{eq:f_weight}
    f(K_x,K_{\phi},{E^x},{E^{\phi}} ) = f( \widetilde{K}_x,K_{\phi},{E^x}) \quad{\rm with}\quad \widetilde{K}_x = K_x/{E^{\phi}}\,.
\end{eqnarray}

\noindent In general, as discussed above the chosen polymerization can be linked to some mimetic model. If that mimetic model does not deparametrize, compared to the classical unpolymerized model, the polymerization will introduce an anomaly to the algebra generated by the $C^{\Delta}$ as well. However, as the algebra of $C^{\Delta}$ comes from an intrinsically consistent covariant generalization of GR, in terms of the primary Hamiltonian $H_P^{\Delta}$, this may not introduce any trouble to general covariance as shown in \cite{BenAchour:2017ivq,Han:2022rsx}, but it implies that the Hamiltonian density is not a conserved quantity in that mimetic model. 
\begin{lemma}
\label{lemma:f1_CC}
~\\
Assuming that the polymerized gravitational contribution to the Hamiltonian constraint denoted by $C^{\Delta}$ in \eqref{eq:defpolyhamiltonianconstraint} is, for solutions of the diffeomorphism constraint $C_x$, conserved under the 
    primary Hamiltonian $ H_P^{\Delta}[N^x=0]$, that is
\begin{eqnarray}\label{eq:cdeltaconserved}
    \eval{\poissonbracket{H_P^{\Delta}[N^x=0]}{C^{\Delta}(y)}}_{C_x = 0} = 0 \, ,
\end{eqnarray}
then the most generic form of polymerization function $f$ encoding the holonomy corrections in \eqref{eq:defpolyhamiltonianconstraint} is given by 
\begin{eqnarray}\label{eq:no_kx_f1}
    f = \frac{{f}^{(1)}(K_{\phi},E^x)- {f}^{(2)}(K_{\phi},E^x) K_{\phi}}{(K_{\phi}+4 {E^x} \widetilde{K}_x) K_{\phi}} + \frac{{f}^{(2)}(K_{\phi},E^x)}{ K_{\phi}} -1 \, ,
\end{eqnarray}
with phase space functions ${f}^{(1)},{f}^{(2)}$ satisfying the following partial differential equation
\be\label{final_con_clo}
\frac{h_1 - 2 E^x \partial_{E^x} h_2}{h_2}= \frac{-4 {E^x} \partial_{E^x}{f}^{(2)}+  \partial_{K_{\phi}}{f}^{(1)}}{2 {f}^{(2)}} \eqqcolon \text{Con}_{f},
\ee 
where $h_1,h_2$ denote the inverse triad corrections from \ref{eq:InvTriadCorr}.
\end{lemma}
\begin{proof}
    The proof is given in Appendix \ref{app:bracketpolyHamiltonconstraint}.
\end{proof}
\noindent Plugging the complicated version of the polymerization function $f$ encoding the holonomy corrections given in equation \ref{eq:no_kx_f1} into the polymerized gravitational Hamiltonian constraint defined in equation \ref{eq:defpolyhamiltonianconstraint}, we can see that the curvature terms assume the form
\begin{eqnarray}
   C^{\Delta}(x)= \frac{E^{\phi}}{2 G E^x}\left[ -\sqrt{{{E^x}}} \qty(\frac{4 K_x {f}^{(2)}(K_{\phi},E^x)}{E^{\phi}} + \frac{{f}^{(1)}(K_{\phi},E^x)}{E^x} ) + h_1 \dots \right](x)\,.
\end{eqnarray}
In other words the polymerization of $K_{x}$ is completely removed due to the requirement of the  closure of the constraint algebra. Using equation \eqref{eq:commute_cx} and from previous lemma \eqref{eq:cdeltaconserved}, we have for arbitrary $N^x$
\begin{eqnarray}
     \eval{\poissonbracket{H_P^{\Delta}[N^x]}{C^{\Delta}(y) }}_{C^{\Delta} = C_x = 0}  = 0\,,
\end{eqnarray}
namely $C^{\Delta} =0$ is a conserved quantity for diffeomorphism invariant solutions thus vacuum solution is generally allowed.
\noindent The exact result of the Poisson bracket between the two polymerized Hamiltonian constraint is
\be
\poissonbracket{C^{\Delta}[N_1]}{C^{\Delta}[N_2]} = \Big( \qty(\partial_{K_{\phi}} {f}^{(2)}) \frac{E^x }{(E^{\phi})^2} C_x\Big) [N_1 N_2' - N_2 N_1'] \,.
\ee 
The details can be found in Appendix \ref{app:bracketpolyHamiltonconstraint}. This is identical to the classical result up to the deformation function $\partial_{K_{\phi}} {f}^{(2)}$ that depends on the given form of the polymerization function still encoded in the choice of ${f}^{(2)}$.

\subsection{Compatible LTB conditions}\label{sec:LTBcondition}
In this subsection we want to shift our focus to the implementation of a generalized LTB condition for the effective system. As we want to consider LTB conditions for polymerizations that correspond to mimetic models that either do deparametrize or do not, we will not directly build on all the results from the previous subsection, due to the restrictive nature of implications for the form of holonomy correction function $f$ stemming from the requirement of a closed algebra of the polymerized geometric contributions to the Hamiltonian constraint. Here we want to relax this a bit by only demanding the conservation of the diffeomorphism constraint $C_x$. This means we will work with the primary Hamiltonian in \eqref{eq:defprimarypolyHamiltonian}, where the gravitational contribution to the Hamiltonian constraint is given in \eqref{eq:defpolyhamiltonianconstraint}. Furthermore we assume that  the polymerization  function $f$ encoding the holonomy corrections fulfills equation \eqref{eq:f_weight}.
~\\
~\\
Instead of using the constraint analysis as described in section \ref{sec:ClassMod} in the classical theory, we would like to work directly with the equations of motion in this subsection. The existence of a compatible LTB condition with $N^x = 0$ implies for such type of solutions, that the four first order equations of motion for $K_x, E^x, K_{\phi}, E^{\phi}$ can be reduced into two independent equations provided the diffeomorphism $C_x = 0$ is satisfied. The set of independent equations of motion will be the same as the one obtained from Dirac's algorithm implementing the appropriate gauge fixing. The reason for choosing this approach is that the equations have a much simpler structure and as a result algebraic manipulations and numerical analysis are less complicated. Inspired by what happens in the classical LTB condition, we would like to represent the independent equations using $E^x, K_{\phi}$. We first solve for $K_x$ using the classical diffeomorphism $C_x$. Then we introduce a change of variable from $E^{\phi}$ to ${g}$ with ${g} \coloneqq \frac{ \partial_x {E^x}}{2 E^{\phi}}$ which is inspired by classical LTB conditions in \eqref{eq:EandKLTBreductionnonmarginal}. Using the new variables $E^x, K_{\phi}, {g}$, where $K_x,E^{\phi}$ are now given by
\begin{eqnarray}\label{eq:Ep_Kx_ltb}
      E^{\phi} = \frac{ \partial_x {E^x}}{2 {g}}, \quad K_x = \frac{ \partial_x {K_{\phi}}}{2 {g}} , \qquad \widetilde{K}_x \coloneqq \frac{K_x}{E^\phi}= \frac{ \partial_x {K_{\phi}}}{\partial_x {E^x}}\, ,
\end{eqnarray}
the equations of motion are equivalent to 
\begin{eqnarray}
\label{eq:Kp_ltb_eom}
    \partial_t K_{\phi} &=& \frac{{g}^2( 2 h_2 + 4 E^x \partial_{E^x} h_2 -h_1) -h_1}{2 \sqrt{E^x}} + \mathcal{F}_{K_\phi}(K_{\phi}, E^x, \partial_x K_{\phi}, \partial_x {E^x}) ,%
    \\
    \label{eq:Ex_ltb_eom}
    \partial_t E^x  &=& \mathcal{F}_{E^x}(K_{\phi}, E^x, \partial_x K_{\phi}, \partial_x {E^x}) ,%
    \\
    \label{eq:ltb_eom}
    {\partial_t {g}}  &=& {g} \, \mathcal{F}_{{g}}(K_{\phi}, E^x, \partial_x K_{\phi}, \partial_x {E^x}, \partial_x^2 K_{\phi}, \partial_x^2 {E^x}) \,. %
\end{eqnarray}
The explicit forms of $\mathcal{F}_{K_\phi}, \mathcal{F}_{E^x}, \mathcal{F}_{{g}}$ are shown in Appendix \ref{app:eom}. They depend on $K_{\phi},E^x$ and their derivatives only. One can check that the equation of motion for $K_x$ 
is automatically satisfied for $C_x = 0$. This make sense since $C_x$ is a conserved quantity for $H_P^{\Delta}$ as previously stated. We note that the first term on the right hand side of \eqref{eq:Kp_ltb_eom} comes from the variation of the spin connection which was ignored in \cite{Bojowald:2008ja}.
Here $\mathcal{F}_{K_\phi}, \mathcal{F}_{E^x}, \mathcal{F}_{{g}}$ are densities of weight zero, thus the density weight is consistent in (\ref{eq:Kp_ltb_eom}-\ref{eq:ltb_eom}) as ${g}$ has density weight zero from its definition. In the classical theory \eqref{eq:ltb_eom} gives automatically $\partial_t {g} = 0$ for arbitrary initial ${g}$. This indicates that ${g}$ is actually non-dynamical, e.g. ${g} = \ltbf(x)$ due to the classical LTB condition and the system is completely determined by \eqref{eq:Kp_ltb_eom} and \eqref{eq:Ex_ltb_eom}. In general when holonomy and/or inverse triad corrections are present this is no longer the case. Consequently, in order to have \eqref{eq:ltb_eom} as an equation, that is not independent to \eqref{eq:Kp_ltb_eom} and \eqref{eq:Ex_ltb_eom}, the time dependence of ${g}$ must be determined by $K_{\phi},E^x$ and their derivatives. This leads to our definition of a compatible LTB condition:
\begin{definition}\label{def:comp_LTB}
~\\
We say an LTB condition ${g}(t,x)$ is compatible with the effective Hamiltonian $H^\Delta_P$ in \eqref{eq:defprimarypolyHamiltonian} or polymerization encoded in $f,h_1$ respectively, if it satisfies \eqref{eq:ltb_eom} and can be expressed as 
\begin{eqnarray}\label{eq:an_g_eom}
    {g}(t,x) = {g}_{\Delta}(K_{\phi}, E^x, \partial_x K_{\phi}, \partial_x {E^x}, \cdots, \partial_x^n K_{\phi}, \partial_x^n {E^x}, \ltbf)(t,x) \,,
\end{eqnarray}
with a non-negative integer $n \in \mathbb{N}_{+} \cup \{0\}$ where $\partial_x^0 = \text{Id}$, %
and the function $\ltbf(x)$ depends on the radial coordinate $x$ only. If $\partial_{\ltbf} {g} =0$, namely ${g}$ does not depend on $\ltbf$ at all, and thus the $x$ dependence is also encoded completely in the phase space variables $K_{\phi}, E^x$, we call such a case the marginally bound case, otherwise we call it the non-marginally bound case.
\end{definition}
\noindent Substituting the property \eqref{eq:an_g_eom} of a compatible LTB condition into the equation of motion in \eqref{eq:ltb_eom} and successively using the other two equations of motions in \eqref{eq:Kp_ltb_eom} and \eqref{eq:Ex_ltb_eom} the equation \eqref{eq:ltb_eom} can be rewritten in a form that does no longer involve first order temporal derivatives. If we now collect the involved contributions by their order of partial derivatives acting on $K_\phi$ and $E^x$, we will get conditions in the following form\footnote{Here we require $K_{\phi},E^x$ are at least of differentiability class $C^{2n}$ for $x$ and of $C^1$ in $t$}:
\begin{eqnarray}
    F_{C_0} + \sum_{k=1}^{n}F_{C_k}^n(K_{\phi}, E^x, \partial_x K_{\phi}, \partial_x {E^x}, \cdots, \partial_x^n K_{\phi}, \partial_x^n {E^x},{g},\partial {g}, \cdots, \partial^n {g} ) \partial_x^{n+k} K_{\phi} + \nonumber\\ \label{eq:defofFc}
    \sum_{k=1}^{n}F_{D_k}^n(K_{\phi}, E^x, \partial_x K_{\phi}, \partial_x {E^x}, \cdots, \partial_x^n K_{\phi}, \partial_x^n {E^x},{g},\partial {g}, \cdots, \partial^n {g} ) \partial_x^{n+k} E^x =0 \,.
\end{eqnarray}
Since $\partial_x^{n+k} K_\phi$ and $\partial_x^{n+k} E^x$ are not contained in ${g}$ for all finite $n$ except for $n \to \infty$, we have in total $2n + 1$ independent conditions $F_{C_0} = F_{C_k}^n =F_{D_k}^n=0$. Due to the reason that all higher order $F_{C_k}, F_{D_k}$ terms are generated from $\partial_x^k  {g}^2 $ in \eqref{eq:Kp_ltb_eom}, it is easy to see that
\begin{lemma}
\label{lemma:g1_first}
~\\
For all finite $n < \infty$, the only allowed LTB condition compatible with the polymerization function $f$ encoding the holonomy corrections in \eqref{poly_f} will be of the form 
    \begin{eqnarray}\label{eq:an1_g_eom}
            {g}_{\Delta}=  {g}_{\Delta}^{(1)}( K_{\phi}, E^x, \ltbf) + {g}_{\Delta}^{(2)}(\widetilde{K}_x, K_{\phi}, E^x)\,,
    \end{eqnarray}
subject to the following  condition
    \begin{eqnarray}\label{eq:an1_g_eom_con}
        2 \left(2 h_2 + 4 E^x \partial_{E^x} h_2 - h_1 \right) \left(\partial_{\widetilde{K}_x}{g}_{\Delta}\right)^2 = 8 {E^x} K_{\phi} \partial_{\widetilde{K}_x}f + K_{\phi} (K_{\phi} + 4 {E^x}\widetilde{K}_x ) \partial_{\widetilde{K}_x}^2f\,.
    \end{eqnarray}
\end{lemma}

\begin{proof}
~\\
Suppose we have $n>1$ and 
$\partial_t {g} = \partial_t {g}_{\Delta}(t,x)$ in \eqref{eq:ltb_eom}.  Substituting the ansatz \eqref{eq:an_g_eom} into \eqref{eq:ltb_eom} leads to
    \begin{eqnarray}
        &&\partial_t {g} =\sum_{i=0}^n  \Big[ \qty(\partial_{\partial_{x}^i K_{\phi}} {g}_{\Delta}) \partial_t \partial_{x}^i K_{\phi} + \qty(\partial_{\partial_{x}^i E^{x}} {g}_{\Delta}) \partial_t \partial_{x}^i E^{x}  \Big]  \ \nonumber \,, %
    \end{eqnarray}
    with $\partial_t \partial_x^n K_{\phi} = \partial_x^n (\partial_t K_{\phi})$ containing the term $\partial_{\partial_x^n K_{\phi}} {g}_{\Delta} \partial_x^{2n} K_{\phi}$ using \eqref{eq:Kp_ltb_eom}. This will introduce e.g. a $\partial_x^{2n} K_{\phi}$ term for the highest power in $n$ appearing in ${g}_{\Delta}$.
    Since the right hand side of \eqref{eq:ltb_eom} contains only derivatives up to order $n$, this implies $\partial_{\partial_x^n K_{\phi}} {g}_{\Delta} = 0$. With a similar procedure, iteratively we have $\partial_{\partial_x^{i} K_{\phi}} {g}_{\Delta} = 0$ for all integer orders $ 2 \leq i \leq n$. Now using \eqref{eq:Ex_ltb_eom} similarly for contributions from $\partial_x^n E^x$ to $\partial_t {g}_{\Delta}$ we have a term proportional to $\partial_x^{n+1} E^x$:
    \begin{eqnarray}
        && 0=\left( \partial_{\partial_x E^x} \mathcal{F}_{E^x} +\frac{{g}_{\Delta}(2 h_2 + 4 E^x \partial_{E^x} h_2 -h_1) }{ \sqrt{E^x}} \partial_{\partial_x K_x} {g}_{\Delta} \right)\qty(\partial_{\partial_x^{n} E^x} {g}_{\Delta}) \partial_x^{n+1} E^x \, , \nonumber
    \end{eqnarray}
    which must vanish, thus $\partial_{\partial_x^{n} E^x} {g}_{\Delta} = 0$. This implies    
\begin{eqnarray}
        {g}(t,x) = {g}_{\Delta}(K_{\phi}, E^x, \partial_x K_{\phi}, \partial_x {E^x}, \cdots, \partial_x^{n-1} {E^x}, \ltbf)(t,x).
    \end{eqnarray}
    The procedure can be performed iteratively until we reach $n=1$ and we finally end up with
    \begin{eqnarray}
        {g}(t,x) = {g}_{\Delta}(K_{\phi}, E^x, \partial_x K_{\phi}, \partial_x {E^x}, \ltbf)(t,x)\,.
    \end{eqnarray}
It is then straight forward to calculate $F_{C_1}^1,F_{D_1}^1$ which are the coefficients for the terms involving $\partial_x^2$ in \eqref{eq:ltb_eom} as can be seen in equations \ref{eqapp:fc1} and \ref{eqapp:fd1}. Further appendix \ref{app:ltb_general} shows in detail that the solution is given by equation \eqref{eq:an1_g_eom} subject to the  condition shown in \eqref{eq:an1_g_eom_con}. Moreover, since the right hand side of equation \eqref{eq:an1_g_eom_con} does not contain $\ltbf$, this means $\partial_{\widetilde{K}_x}{g}_\Delta$ is independent of $\ltbf$ too. Only in the integration constant ${g}_{\Delta}^{(1)}( K_{\phi}, E^x, \ltbf)$ the LTB function can be present.
\end{proof}
\noindent
We can define the analogue of the classical gauge fixing condition $G_x$ of \eqref{eq:GFconditionsclassical} in the polymerized model by
\begin{equation}
\label{eq:GxDelta}
G_x^{\Delta}(x):= \frac{E^x{}'}{2 E^{\phi}}(x) - {g}_{\Delta}(x)\,,   \end{equation}
with ${g}_{\Delta}$ given by \eqref{eq:an1_g_eom}, whose classical limit is $\lim\limits_{\Delta\to 0} G_x^{\Delta}=G_x$.
Likewise to the classical case, we have that for non-constant ${g}_{\Delta}(x)$, $G_x^{\Delta}$ and the spatial diffeomorphism constraint $C_x$ build a second class constraint pair
\begin{eqnarray}
    \poissonbracket{C_x(x)}{G_x^{\Delta}(y)} = (\partial_{\ltbf}{g}_{\Delta}) \ltbf'(x) \delta(x,y)\,.
\end{eqnarray}
Now with equation \eqref{eq:an1_g_eom}, the only remaining condition in \eqref{eq:defofFc} is $F_{C_0}=0$ and we can rewrite it as
\begin{eqnarray}\label{eq:LTB_gene}
    \frac{(2 h_2 + 4 E^x \partial_{E^x} h_2 -h_1){g}_{\Delta}}{\partial_x E^x} \qty(\partial_{\ltbf} {g}_{\Delta}) \qty(\partial_{\widetilde{K}_x} {g}_{\Delta}) \ltbf' + \frac{1}{4 (E^x)^{3/2}}\widetilde{F}_{C_0}(\widetilde{K}_x,K_{\phi},E^x,{g}_{\Delta},\partial {g}_{\Delta},f,\partial f) = 0\,,\nonumber \\
\end{eqnarray}
where the explicit form of $\widetilde{F}_{C_0}$ is given in \ref{eqapp:def_tfc0}. As ${g}_{\Delta}$ does not contain $\ltbf'$, we can separate the solution of ${g}_{\Delta}$ into two classes depending on whether $\partial_{\ltbf} {g}_{\Delta}$ vanishes or not as the term $(2 h_2 + 4 E^x \partial_{E^x} h_2 -h_1)$ has a non vanishing classical limit and thus is not allowed to vanish identically in the effective theory. For the first option this is the marginally bound case and and the second possibility covers the non-marginally bound case. 
This result can be summarized in the following Lemma.
\begin{lemma}
\label{thm:marginal}
\textbf{(marginally bound case)} ~\\
A compatible LTB condition in the marginally bound case in the sense of Definition \ref{def:comp_LTB} is determined by the following $f=f(\widetilde{K}_x,K_{\phi},E^x)$ and ${g}_{\Delta} = {g}_{\Delta}^{(2)}(\widetilde{K}_x,K_{\phi},E^x)$ which are given by the solutions of the following equations
\begin{eqnarray}%
    0 &=& 8 {E^x} K_{\phi} \partial_{\widetilde{K}_x}f + K_{\phi} (K_{\phi} + 4 {E^x}\widetilde{K}_x) \partial_{\widetilde{K}_x}^2f- 2 \left(2 h_2 + 4 E^x \partial_{E^x} h_2 -h_1 \right) \left(\partial_{\widetilde{K}_x} {g}_{\Delta}\right)^2 \\ %
     0 &=& \widetilde{F}_{C_0}(\widetilde{K}_x,K_{\phi},E^x,{g}_{\Delta},\partial {g}_{\Delta},f,\partial f)
    \end{eqnarray}
\end{lemma}
\begin{proof}
The proof is straightforward by using Lemma \ref{lemma:g1_first} and equation \eqref{eq:LTB_gene} for the special case  $\partial_{\ltbf} {g}_{\Delta} = 0$
\end{proof}
\noindent
Note that the $K_x$ dependence of the polymerization function $f$ encoding the holonomy corrections and ${g}_{\Delta}$ will be determined simultaneously, they are not independent. This means for a given specific $f$ with $K_x$ corrections, it might happen that no solution exists for ${g}_{\Delta}$. We can make this precise for the non-marginally bound case in the following Lemma.
\begin{lemma}\label{thm:non_marginally bound case }\textbf{(non-marginally bound case)}
~\\
To have a compatible LTB condition ${g}_{\Delta}$ in the non-marginally bound case in the sense of Definition \ref{def:comp_LTB}, the polymerization function $f$ encoding the holonomy corrections as well as ${g}_{\Delta}$ are restricted to the form of
      \begin{eqnarray}\label{eq:no_kx_f1_g1_n1}
    f = \frac{{f}^{(1)}(K_{\phi},E^x)}{K_{\phi}+4 {E^x} \widetilde{K}_x} + {f}^{(2)}(K_{\phi},E^x) \ , \qquad
    {g}_{\Delta} =   {g}_{\Delta}^{(1)}( K_{\phi}, E^x, \ltbf).
    \end{eqnarray}
Moreover, they satisfy
    \begin{eqnarray}\label{eq:con_non_m_g1}
            0 = \mathcal{G}_1[{g}_{\Delta},f]+\mathcal{G}_2[{g}_{\Delta},f]
    \end{eqnarray}
    where $\mathcal{G}_{1}$ is linear and $\mathcal{G}_2$ cubic in ${g}_{\Delta}$ (and its derivatives $\partial_{K_{\phi}} {g}_{\Delta}$).  The explicit forms of  $\mathcal{G}_{1}$ and $\mathcal{G}_2$ read
\begin{equation}\label{eq:G12_expan_explicit}
    \begin{split}
           \mathcal{G}_1 =& 2 {f}^{(2)} \left( {g}_{\Delta} -2 E^x \partial_{E^x} {g}_{\Delta}\right) +{g}_{\Delta} \left(4 E^x \partial_{E^x} {f}^{(2)} - \partial_{K_{\phi}}{f}^{(1)}\right)
        -(h_1+ {f}^{(1)}) \partial_{K_{\phi}} {g}_{\Delta}\\
        \mathcal{G}_2 =&({g}_{\Delta})^2 \left(2 h_2 + 4 E^x \partial_{E^x} h_2 -h_1\right) \partial_{K_{\phi}} {g}_{\Delta}\, .
    \end{split}
\end{equation}
\end{lemma}
\begin{proof}~\\
In the non-marginally bound case we have $\partial_{\ltbf} {g}_{\Delta} \neq 0$ and thus $\partial_{\widetilde{K}_x} {g}_{\Delta}=0$ due to equation \eqref{eq:LTB_gene} and the argument presented right below \eqref{eq:LTB_gene}. Plugging this in the condition \eqref{eq:an1_g_eom_con} from Lemma \ref{lemma:g1_first} we can derive the general solution for holonomy corrections which is exactly the first equation in \eqref{eq:no_kx_f1_g1_n1}. Substituting this back into the equation  \eqref{eq:LTB_gene} gives condition \eqref{eq:con_non_m_g1}.
\end{proof}
\noindent Note that clearly \eqref{eq:con_non_m_g1} is also the condition in the marginally bound case when the polymerization function $f$ encoding the holonomy corrections takes the form in \eqref{eq:no_kx_f1_g1_n1}, or in other words when there is no polymerization of $\widetilde{K}_x$ as stated in Corollary 1 below.
\begin{corollary}\label{cor:gEx_m}
~\\
A compatible LTB condition ${g}_{\Delta} =   \tilde{g}_{\Delta}( K_{\phi}, E^x) $ in the marginally bound case exists when the polymerization function $f$ encoding the holonomy corrections is given by \eqref{eq:no_kx_f1_g1_n1}, namely there is no polymerization of $\widetilde{K}_x$ involved. The compatible LTB condition 
     ${g}_\Delta =   \tilde{g}_{\Delta}( K_{\phi}, E^x) $ is determined by the solution of \eqref{eq:con_non_m_g1}.
\end{corollary}

\noindent Another interesting case is for  non-marginally bound  solutions when we assume the form of the LTB condition to be ${g}_{\Delta} = \tilde{g}_{\Delta}( K_{\phi}, E^x) \ltbf$. Then $\mathcal{G}_1$ and $\mathcal{G}_2$ contain different orders of ${g}$ and $\ltbf$. Hence the condition in \eqref{eq:con_non_m_g1} actually implies $\mathcal{G}_1 = \mathcal{G}_2 =0$. This is the content of Corollary \ref{cor:gEx_non_m}.
\begin{corollary}
\label{cor:gEx_non_m}
~\\
For an LTB condition of the form ${g}_{\Delta} = \tilde{g}_{\Delta}( K_{\phi}, E^x) \ltbf$ in the non-marginally bound case, the compatible LTB condition is given by
    \begin{eqnarray}    \label{eq:ltb_con_to_f1_f2}
    {g}_\Delta = \tilde{g}_{\Delta}(E^x)\ltbf \, ,\qquad
    1-\frac{2{E^x} \partial_{E^x} \tilde{g}_{\Delta}}{\tilde{g}_{\Delta}} =\frac{-4 {E^x} \partial_{E^x}{f}^{(2)}+  \partial_{K_{\phi}}{f}^{(1)}}{2 {f}^{(2)}} =\text{Con}_{f}%
\end{eqnarray}
where $\text{Con}_{f}$ must be a function that depends on $E^x$ only, which puts even more constraints on the polymerization function $f$ encoding the holonomy corrections.
\end{corollary}
\noindent Note that the condition $\text{Con}_{f}$ on the right hand side of the second equation in \eqref{eq:ltb_con_to_f1_f2} is responsible for closing the algebra between $C^{\Delta}$ and $C_x$ as shown in Lemma \ref{lemma:f1_CC}. This makes the polymerization function $f$ encoding the holonomy corrections given in the form of Lemma \ref{lemma:f1_CC} a special class of models for which the algebra of $C^{\Delta}$ and $C_x$ is closed. Note that the algebra of the total constraints can close even if the algebra of $C^{\Delta}$ and $C_x$ does not. In this case the polymerization corresponds to mimetic models that do not deparametrize. A special case is the vacuum case where $C^{\Delta}$ itself is a constraint again. The vacuum case will be discussed in a separate article \cite{GenBirkhoff}.
~\\
~\\
For the class of  models where the algebra of $C^{\Delta}$ and $C_x$ closes, we can formulate the following two further corollaries. The first one, Corollary \ref{cor:fEx_clo} discussed the existence and form of compatible LTB conditions in the class of models that are covered by Lemma  \ref{lemma:f1_CC}.
\begin{corollary}
\label{cor:fEx_clo}
~\\
A compatible LTB condition ${g}_{\Delta} = \tilde{g}_{\Delta}(E^x)\Xi$ in both the marginally and non-marginally bound case always exists if the polymerization function encoding the holonomy corrections $f$ is given by the form in Lemma \ref{lemma:f1_CC}. In this case $C^{\Delta}$ in \eqref{eq:defpolyhamiltonianconstraint}
is conserved under the dynamics and there is a closed algebra of $C^{\Delta}$ and $C_x$. A solution for ${g}_{\Delta} = \tilde{g}_{\Delta}(E^x) \ltbf$ is given by
    \begin{eqnarray}\label{eq:g_as_fEx}
         {2{E^x} \partial_{E^x} \tilde{g}_{\Delta}} = \left(1-\frac{h_1 - 2 E^x \partial_{E^x} h_2}{h_2} \right)\tilde{g}_{\Delta}
    \end{eqnarray}
    Moreover, the classical LTB condition ${g}_c=\Xi(x)$ is compatible if and only if further requiring that the functions $h_1,h_2$ encoding the inverse triad corrections  in \eqref{eq:InvTriadCorr} satisfy $h_1 = h_2 + 2 E^x \partial_{E^x} h_2$. 
\end{corollary}
\noindent A special case of Corollary \ref{cor:fEx_clo} is when $h_1 =h_2=1$, such that there are no inverse triad corrections present. Then  as can be easily seen from \eqref{eq:g_as_fEx} the classical LTB condition ${g}_c=\Xi(x)$ is a compatible one.
~\\
Note that in \cite{Bojowald:2005cb,Bojowald:2009ih} only models with a polymerized LTB condition were considered and the class of models with $h_1 =h_2=1$ were not considered.
Furthermore, as also discussed more in detail in section \ref{sec:examples} the closure of the algebra of $C^{\Delta}$ and $C_x$ and thus the conservation of the scalar density, was not a criterion for the models considered in \cite{Bojowald:2005cb,Bojowald:2009ih}.
~\\
~\\
In the final Corollary of this subsection the form of the polymerized scalar constraint and effective dynamics are discussed. In addition it is also shown how these dynamics can be linked back to the Dirac brackets discussed in subsection \ref{sec:ClassNonMarg}.
\begin{corollary}
\label{cor:C_and_dynamics}
~\\
Models with a polymerization function $f$ encoding the holonomy corrections given by the form in Lemma \ref{lemma:f1_CC} and further the LTB condition ${g}_{\Delta} = \tilde{g}_{\Delta}(E^x)\Xi$ due to Corollary \ref{cor:gEx_non_m} and \ref{cor:fEx_clo} yield an effective scalar density of the form
    \begin{eqnarray}\label{eq:def_C_tilde}
        C^{\Delta} = \frac{1}{2G}\left( \frac{\ltbf'}{\ltbf} + \partial_x \right)\widetilde{C}^{\Delta} \qquad{\rm with}\quad  \widetilde{C}^{\Delta} = \frac{\sqrt{E^x}}{\tilde{g}_{\Delta} \ltbf} \left( - F^{(2)} + h_2 \left( \tilde{g}_{\Delta}^2 \ltbf^2 -1 \right) \right) 
    \end{eqnarray}
    where we defined $\partial_{K_{\phi}} F^{(2)}(K_{\phi},E^x) =2 {f}^{(2)}(K_{\phi},E^x)$ satisfying 
\begin{eqnarray}
        \frac{h_1 - 2 E^x \partial_{E^x} h_2}{h_2} = \frac{-2 E^x \partial_{E^x} F^{(2)} + f^{(1)} }{ F^{(2)} }\, .
    \end{eqnarray}
    The effective dynamics in (\ref{eq:Kp_ltb_eom}-\ref{eq:Ex_ltb_eom}) reduce for those class of models to
    \begin{eqnarray}\label{eq:eff_dy_new}
        \partial_t E^x = 2 \sqrt{E^x} {f}^{(2)} \ , \qquad
        \partial_t K_{\phi} = - \frac{1}{2 \sqrt{E^x}} \left( f^{(1)}-  \tilde{g}_{\Delta}^2 \ltbf^2 (2 h_2 + 4 E^x \partial_{E^x} h_2 -h_1) + h_1 \right).
    \end{eqnarray}
These equations are valid for the marginally as well as non-marginally bound case. In case of the latter they can be generated from $C^{\Delta}$ with the Dirac bracket corresponding to the second class constraint pair $G_x^\Delta$ and $C_x$
    \begin{eqnarray}\label{eq:eff_dirac_b}
        \poissonbracket{K_{\phi}(x)}{E^x(y)}_D = - 2 G \frac{ \ltbf^2 }{\ltbf'}\tilde{g}_{\Delta} \delta(x-y).
    \end{eqnarray}
    One of the properties of these class of models is that the effective dynamics decouples completely along the radial $x$ direction. Clearly $\widetilde{C}^{\Delta}$ is conserved under the effective dynamics.
\end{corollary}
~\\
The fact that the effective dynamics are decoupled completely along the radial direction is interesting also in the context of the dust shell models considered in \cite{Kiefer:2019csi,Giesel:2021dug} where such a decoupling was one of the basic assumptions for the dust shell models. The work in this article shows that under the further assumptions discussed in Corollary \ref{cor:C_and_dynamics} this is indeed justified.
~\\
Note that for monotonic $\ltbf(x)$, %
we can rewrite \eqref{eq:eff_dirac_b} as
\begin{eqnarray}\label{eq:Diracbracket_simplified}
    \poissonbracket{K_{\phi}(\ltbf^{-1}(x))}{E^x(\ltbf^{-1}(y))}_D =  2 G \,\mathrm{sgn}\qty(\ltbf'(x))\,   \tilde{g}_{\Delta} \delta\left(\ltbf^{-1}(x)-\ltbf^{-1}(y)\right)
\end{eqnarray}
with the primary Hamiltonian up to boundary terms given by
\begin{eqnarray}
    H^{\Delta}_{\text{pri}} = \frac{1}{2G} \int \text{d}x \  \frac{\ltbf'}{\ltbf^2} \widetilde{C}^{\Delta} = \frac{1}{2G} \int \text{d} \ltbf^{-1} \ \widetilde{C}^{\Delta} 
\end{eqnarray}
Given this result, one can then either use this as a classical starting point and proceed with the standard loop quantum gravity inspired quantization procedure or use it as an effective model corresponding to such quantum models.
~\\
As discussed in subsection \ref{sec:ClassNonMarg} coming from a spherically symmetric model and 
implementing the LTB condition for the non-marginally bound case can be understood as a gauge fixing with respect to the spatial diffeomorphism constraint. Hence, along the lines of the work in \cite{Giesel:2021rky} we can address the question for which class of models gauge fixing and quantization commute. Here this carries over to implementing the LTB condition at the classical level and quantize/polymerize afterwards or first polymerize and afterwards implement the LTB condition with respect to the effective dynamics. In particular the difference manifests in models where either the whole spherically symmetric sector is polymerized/quantized or the LTB sector only.  In general we do not expect such a commutativity between the two steps. However, there exist a class of models where these two steps indeed commute: 
\begin{remark}
~\\
For non-marginally bound models for which the compatible LTB condition is exactly the classical one, that is $G_x^\Delta$  in \eqref{eq:GxDelta} exactly agrees with $G_x$ in \eqref{eq:GFconditionsclassical},
polymerization and implementing the LTB condition commute. 
\end{remark}

\section{Applications}\label{sec:examples}
On the basis of the general results given in previous section, we want to discuss in the following two concrete examples. The first one in subsection \ref{sec:example_1} corresponds to the models considered in the context of work in \cite{Bojowald:2008ja} and \cite{Bojowald:2009ih}. The first paper mainly considered models in the marginally bound case, whereas the second one extends the analysis to the non-marginally bound sector. These models do not have a polymerization of the radial curvature component $K_x$, but go beyond the scope of theories described in Lemma \ref{lemma:f1_CC} and thus are in accordance with the results from Corollaries \ref{cor:gEx_m} and \ref{cor:gEx_non_m}. We will show a detailed comparison between the results given in \cite{Bojowald:2008ja} and \cite{Bojowald:2009ih} and study their dynamics numerically. 
~\\
The second example discussed in subsection \ref{sec:example_2} is in the class of models covered by Lemma \ref{lemma:f1_CC} with a $\frac{\sin(\mu x)}{\mu}$-polymerization in both the $\mu_0$ and the $\bar{\mu}$ schemes with no inverse triad corrections present. The latter effective model can be found in the work \cite{Tibrewala:2012xb}. Applying Corollaries \ref{cor:fEx_clo} and \ref{cor:C_and_dynamics}, we can see that in this example for both schemes the classical LTB condition is the compatible LTB condition. In addition their equations of motion decouple at each $x$ and coincide with the LQC effective equations for each $x$. 

\subsection{Example 1}
\label{sec:example_1}
Based on the models in \cite{Bojowald:2008ja} and \cite{Bojowald:2009ih}, we want to analyze effective LTB models that have holonomy corrections of the form $K_{\phi} \to \sin(\beta \sqrt{\Delta} K_{\phi})/(\beta \sqrt{\Delta})$ as well as inverse triad corrections. Any other holonomy corrections of $K_{\phi}$ can be treated similarly. Note that $K_x$ is not polymerized in these models. In the following we will consider inverse triad corrections switched on and off. In our notation, the holonomy corrections are parameterized by a function $f$ (see equation \eqref{eq:defpolyhamiltonianconstraint} and \eqref{eq:no_kx_f1}) comparing this with equation (75) and (49) as well as (50) for the two types of inverse triad corrections in \cite{Bojowald:2008ja} or with equations (30) and (28) in \cite{Bojowald:2009ih}, we obtain the following relations with the notation used in \cite{Bojowald:2008ja,Bojowald:2009ih}
\begin{eqnarray}
  && \text{I} :  \qquad {f}^{(1)} = h_1 ({f}^{(2)})^2 \ , \quad {f}^{(2)} = \frac{\sin(\beta \sqrt{\Delta} K_{\phi})}{\beta \sqrt{\Delta}} \ , \quad h_2 = l(E^x) \ , \\
  && \text{II} :  \qquad {f}^{(1)} = ({f}^{(2)})^2/h_1 \ , \quad {f}^{(2)} = h_1 \frac{\sin(\beta \sqrt{\Delta} K_{\phi})}{\beta \sqrt{\Delta}} \ , \quad h_2 = h(E^x)\ ,
\end{eqnarray}
with $\Delta = 4 \pi l_p^2 $ and the inverse triad corrections $h_1$ are taken as
\begin{eqnarray}
    h_1(E^x) = \alpha(E^x) =\sqrt{E^x} \frac{\sqrt{E^x + \beta l_p^2/2} - \sqrt{E^x - \beta l_p^2/2}}{\beta l_p^2/2} \,. 
\end{eqnarray}
Case $\text{I}$ and $\text{II}$ refer to the two types of inverse triad corrections introduced in \cite{Bojowald:2008ja,Bojowald:2009ih}. The second function $h_2(E^x)$ encoding inverse triad corrections is only necessary to parameterize corrections related to the spin connection encoded in
$l(E^x),h(E^x)$ introduced in \cite{Bojowald:2009ih}. It is then straightforward to check that 
\begin{eqnarray}
    && \text{I} :  \qquad \frac{-4 {E^x} \partial_{E^x}{f}^{(2)}+  \partial_{K_{\phi}}{f}^{(1)}}{2 {f}^{(2)}} = \cos(\beta \sqrt{\Delta} K_{\phi}) h_1 \\
     && \text{II} :  \qquad \frac{-4 {E^x} \partial_{E^x}{f}^{(2)}+  \partial_{K_{\phi}}{f}^{(1)}}{2 {f}^{(2)}} = \frac{1}{h_1} \left( \cos(\beta \sqrt{\Delta} K_{\phi}) h_1 - 2 E^x \partial_{E^x} h_1 \right).
\end{eqnarray}
As the above equations show in both cases, the equations ${f}^{(1)} $ and ${f}^{(2)}$ need to satisfy are not functions of $E^x$ only when there are holonomy corrections. %
Thus we have 
\begin{equation*}
\poissonbracket{H_{P}^{\Delta}}{C^{\Delta}}\eval_{C_x = 0} \neq 0    
\end{equation*}
according to Lemma \ref{lemma:f1_CC} (\eqref{final_con_clo} in Lemma \ref{lemma:f1_CC} is violated), i.e. the Hamiltonian density is not conserved in contrast to the classical theory. This implies there will be no compatible LTB condition in the form of ${g}_{\Delta} = \tilde{g}_{\Delta} \ltbf$ in the non-marginally bound case due to Corollary \ref{cor:gEx_non_m}. Moreover, as shown in \cite{GenBirkhoff}, this further implies there exist no non-trivial stationary vacuum solution. Nonetheless in the marginal bound case a compatible LTB condition exists since we clearly have $\partial_{\widetilde{K}_x} {g}_\Delta = 0$ and thus Corollary \ref{cor:gEx_m} applies. 
~\\
In case there are no holonomy corrections, the existence of a non-trivial $h_2$ can close the algebra of $C^\Delta$ and $C_x$ (see  \eqref{final_con_clo} in Lemma  \ref{lemma:f1_CC}) and so a compatible LTB condition in both the marginally and non-marginally bound case can exist. This is the reason why in \cite{Bojowald:2009ih} the correction function $h_2$ is necessary. However we obtain a different $h_2$ function compared with  \cite{Bojowald:2009ih} in which the classical LTB condition is used to reduce the spin connection as discussed in section \ref{sec:CompBojowald}.
~\\
The modified LTB condition ${g}_\Delta={g}_{\Delta}^{(1)}(K_{\phi},E^x)$ is given by the solution of equation \eqref{eq:con_non_m_g1}, which in this model here becomes
\begin{eqnarray}
   \text{I}: \quad  0&=& \frac{\sin (\beta \sqrt{\Delta} K_{\phi}) }{\beta \sqrt{\Delta} } \left({g}_{\Delta}^{(1)} \left(-1+h_1 \cos (\beta \sqrt{\Delta} K_{\phi})\right)+2 {E^x} \partial_{E^x}{g}_{\Delta}^{(1)}\right)\\
    &&-\left(\left(h_1+({g}_{\Delta}^{(1)}){}^2 (h_1-2 h_2 - 4 E^x \partial_{E^x} h_2)\right)+h_1 \frac{\sin ^2(\beta \sqrt{\Delta} K_{\phi})}{\beta^2 \Delta}\right) \partial_{K_{\phi}}{g}_{\Delta}^{(1)} \nonumber \\
\text{II}: \quad  0&=& 2 \frac{\sin (\beta \sqrt{\Delta} K_{\phi}) }{\beta \sqrt{\Delta} } \left({g}_{\Delta}^{(1)} h_1 \left(-1+\cos (\beta \sqrt{\Delta} K_{\phi})\right)-2 h_1 {E^x} \partial_{E^x} {g}_{\Delta}^{(1)} +2 {g}_{\Delta}^{(1)} {E^x} \partial_{E^x}  h_1\right)\\
    &&-\left(\left(h_1+({g}_{\Delta}^{(1)}){}^2 (h_1-2 h_2 - 4 E^x \partial_{E^x} h_2)\right)+h_1 \frac{\sin ^2(\beta \sqrt{\Delta} K_{\phi})}{\beta^2 \Delta }\right) \partial_{K_{\phi}}{g}_{\Delta}^{(1)}\,. \nonumber
\end{eqnarray}
The term $\left(h_1+({g}_{\Delta}^{(1)}){}^2 (h_1-2 h_2 - 4 E^x \partial_{E^x} h_2)\right)$ is coming from the contribution of the spin connection, which is missing in \cite{Bojowald:2008ja}, see the  definition of Hamiltonian density in (36), as well as to the polymerized ones in (49),(50) and (51).
When the polymerization of $K_{\phi}$ is turned off, above equation reproduces exactly the equation (56) and (71) of \cite{Bojowald:2008ja} for the LTB condition with inverse triad corrections as in such a case $\partial_{K_{\phi}}{g}_{\Delta}^{(1)} =0$ and thus the contribution from the spin connection part can be ignored. 
However, this is not the case when we have holonomy polymerizations. Comparing to equation (79) in \cite{Bojowald:2008ja}, we see the contribution of an extra term $\left(1 - ({g}_{\Delta}^{(1)}){}^2\right)$ by taking $h_1 =h_2 =1$ following our analysis. Moreover, due to the appearance of the term involving $({g}_{\Delta}^{(1)}){}^2$, we do not have an analytic solution to this equation. Therefore, we study this model numerically. Here we restrict our analysis to the stationary solution where $\partial_t + \partial_x =0$. As a result, all phase space functions are only a function of $z=x-t$, more specifically, ${g}_{\Delta}^{(1)} $ is also a function of $z$. Using the \eqref{eq:ltb_eom} %
the equation of motion of the LTB condition can be written as
\begin{eqnarray}
 \text{I}: \quad  \partial_z {g}_{\Delta}^{(1)}(z) &=& \left( \frac{-1+h_1 \cos (\beta \sqrt{\Delta} K_{\phi}) }{ \beta \sqrt{\Delta} \sqrt{E^x}} \sin (\beta \sqrt{\Delta} K_{\phi}){g}_{\Delta}^{(1)} \right)(z) \\
 \text{II}: \quad  \partial_z {g}_{\Delta}^{(1)}(z)  &=& \left( \frac{h_1 \left(-1+\cos (\beta \sqrt{\Delta} K_{\phi})\right) - 2 E^x \partial_{E^x} h_1  }{ \beta \sqrt{\Delta} \sqrt{E^x}} \sin (\beta \sqrt{\Delta} K_{\phi}) {g}_{\Delta}^{(1)} \right)(z)\,.
\end{eqnarray}
It is clear that the equations are invariant under the scaling of ${g}_{\Delta}^{(1)} $. The remaining equations of motion take the form
\begin{eqnarray}
 \text{I} \; \& \; \text{II}: \quad  \partial_z  K_{\phi} &=&  - \frac{ h_1 \sin^2 ( \beta \sqrt{\Delta} K_{\phi})}{2 \beta^2 \Delta \sqrt{{E^x}}}+ \frac{ h_1+({g}_{\Delta}^{(1)})^2(h_1 - 2 h_2 - 4 E^x \partial_{E^x} h_2)}{2 \sqrt{{E^x}}}\\ 
 \text{I}: \quad \partial_z  {E^x} &=& -\frac{2 \sqrt{{E^x}} \sin (\beta \sqrt{\Delta} K_{\phi})}{\beta \sqrt{\Delta}} \\
 \text{II}: \quad \partial_z  {E^x} &=& -\frac{2 h_1 \sqrt{{E^x}} \sin (\beta \sqrt{\Delta} K_{\phi})}{\beta \sqrt{\Delta}}\,.
\end{eqnarray}
We will now solve these differential equations numerically for given initial conditions. Note that the initial value of ${g}_{\Delta}^{(1)} $ is chosen such that when $z\to \infty$ we recover the classical value. The remaining variables $K_x, E^{\phi}$ are then determined by the modified LTB condition and the diffeomorphism constraint ($C_x =0$) as
\begin{eqnarray}
    K_x = \frac{ \partial_z  K_{\phi} }{2{g}_{\Delta}^{(1)}  } \ , \qquad E^{\phi}= \frac{ \partial_z E^x}{2{g}_{\Delta}^{(1)}  }\,.
\end{eqnarray}
The plot of ${g}_{\Delta}^{(1)} $ as a function of $z$ is shown in figure \ref{fig:g1}. One can see that ${g}_{\Delta}^{(1)} $ has a bounce at the to-be classical singularity. The energy density is given by
\begin{eqnarray}
    \rho = \frac{C^{\Delta}}{E^{\phi} \sqrt{E^x}} = \frac{\left( 1-\cos(\beta \sqrt{\Delta} K_{\phi}) \right)({g}_{\Delta}^{(1)})^2}{E^x}\,.
\end{eqnarray}
The graph of the energy density scaled by the volume, which is just $C^{\Delta}= \rho E^{\phi} \sqrt{E^x}$, confirms explicitly that we do not have a vacuum solution in such case, where we refer the reader for more details on the vacuum solution to \cite{GenBirkhoff}. The sign change of $E^{\phi}$ is relevant for the sign change of $C^{\Delta}$. 
\begin{figure}[ht!]
    \centering
    \includegraphics[width=0.45\textwidth]{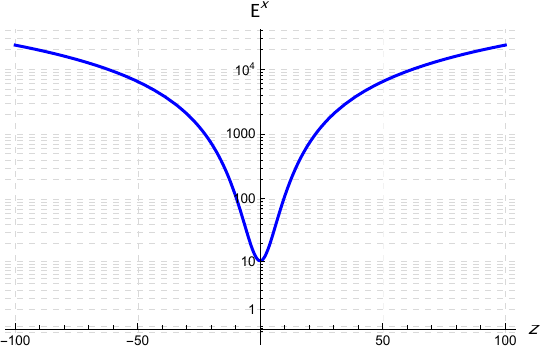}
    \includegraphics[width=0.45\textwidth]{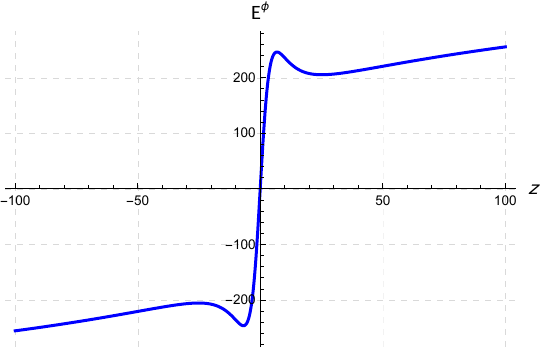}
    \includegraphics[width=0.45\textwidth]{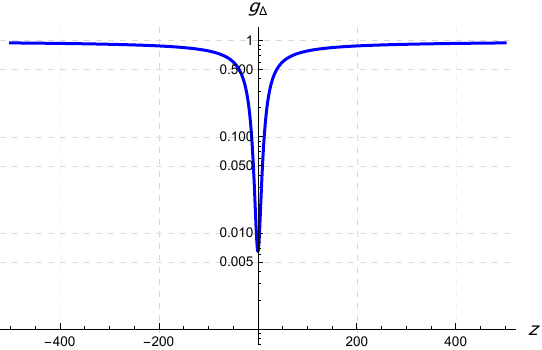}
    \includegraphics[width=0.45\textwidth]{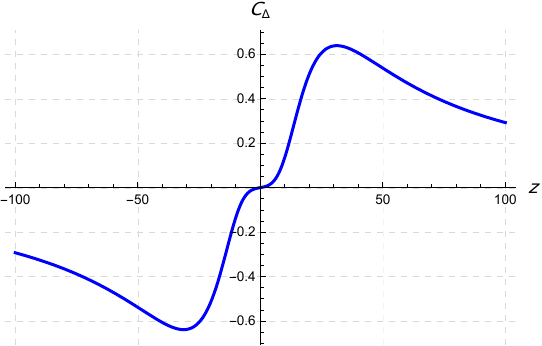}
    \caption{The plot of the $E^x$, $E^{\phi}$, ${g}_\Delta$ and $C_{\Delta}$ under the effective dynamics I. The estimation uses $M = 100, \beta=0.2375$ and we put the classical initial value from classical LTB solutions $E^x = (3 \sqrt{2 M} z_0 /2)^{4/3}, K_{\phi} = -(4 M /3/z_0)^{1/3}, g_{\Delta} = 1$ at $z_0= 10^4$.}
    \label{fig:g1}
\end{figure}

\subsection{Example 2}
\label{sec:example_2}
In the second example found in \cite{Tibrewala:2012xb} we do not consider inverse triad corrections anymore, so we set $h_1 =h_2=1$ and the holonomy corrections are given by the function $f$ in the form of
\begin{eqnarray}
  && {f}^{(1)} = \frac{\sin(\beta \sqrt{\Delta} K_{\phi})^2}{\beta^2 {\Delta}}  \ , \quad {f}^{(2)} = \frac{\sin(2 \beta \sqrt{\Delta} K_{\phi})}{2 \beta \sqrt{\Delta}} 
\end{eqnarray}
We note that the factor of $2$ in the sine function of $ {f}^{(2)}$ is coming from the requirement of the closure of the algebra of $C^\Delta$ and $C_x$ in \eqref{final_con_clo}. At first it may seem that we have not considered a  function preserving polymerization. However, the factor of 2 can be understood as 
\begin{equation*}
\sin(2 \beta \sqrt{\Delta} K_{\phi}) = 2 \sin(\beta \sqrt{\Delta} K_{\phi}) \cos( \beta \sqrt{\Delta} K_{\phi}),  
\end{equation*}
which fulfills the function preserving requirement of holonomy corrections. Due to the reason that equation \eqref{final_con_clo} is satisfied, the system can incorporate the classical LTB condition, where both the marginally bound and the non-marginally bound case are allowed. The compatible LTB condition then reads
\begin{eqnarray}
    {g}_c = \ltbf.
\end{eqnarray}
This model then satisfies exactly the condition imposed in Corollary \ref{cor:fEx_clo} and \ref{cor:C_and_dynamics}. As a result, the equations of motion after imposing the LTB condition are given by
\be\label{eom_mu0_Tibre}
\partial_t E^x &=&  \sqrt{{E^x}} \frac{\sin(2 \beta \sqrt{\Delta} K_{\phi})}{\beta \sqrt{\Delta} }\\
\partial_t K_{\phi} &=&  \frac{\sin(\beta \sqrt{\Delta} K_{\phi})^2}{2 \beta^2 {\Delta}  \sqrt{E^x}} - \frac{\ltbf^2-1}{2 \sqrt{E^x}}.
\ee 
The remaining variables $K_x, E^{\phi}$ are then determined by the LTB condition and the diffeomorphism constraint as
\begin{eqnarray}
    K_x = \frac{ \partial_x  K_{\phi} }{2 \ltbf} \ , \qquad E^{\phi}= \frac{ \partial_x E^x}{2 \ltbf  }.
\end{eqnarray}
The effective geometric contribution to the Hamiltonian constraint becomes
\be\label{eq:Ham_Tibre_mu0}
C^{\Delta}(x) = \frac{1}{2G} \left( \frac{\ltbf'  }{\ltbf} + \partial_x \right) \left( \frac{ \sqrt{E^x}  \sin(\beta \sqrt{\Delta} K_{\phi})^2}{2 \beta^2 {\Delta} \ltbf} - \frac{\sqrt{E^x}}{\ltbf} \left({\ltbf^2-1}\right) \right)(x)
\ee 
according to Corollary \ref{cor:C_and_dynamics}.
Notice that similar to the classical case we see that in the marginally bound case, where $\ltbf'=0$,  we have $C^{\Delta} = 0$. The relation of the marginally bound case to the stationary and vacuum solutions is discussed more in detail in a second article \cite{GenBirkhoff}. Since in this example the LTB condition is exactly the classical one, the Dirac bracket in the effective model exactly agrees with the one in \eqref{eq:DBGravPart}, which means we have
\begin{eqnarray}\label{eq:cano_Tibre}
    \poissonbracket{K_{\phi}(x)}{E^x(y) }_D = - 2G \frac{\ltbf^2}{\ltbf'} \delta(x-y).
\end{eqnarray}
One can check that the equation of motion for $E^x$ in \eqref{eom_mu0_Tibre} can be obtained from the Hamilton's equation defined by the primary Hamiltonian density in \eqref{eq:Ham_Tibre_mu0} (set $N=1$ and $N^x=0$) %
and the bracket in \eqref{eq:cano_Tibre}, which gives a concrete example of Corollary \ref{cor:C_and_dynamics}. 
~\\
~\\
Next we want to construct the $\bar{\mu}$ version of this polymerized model by setting
\begin{eqnarray}
  && {f}^{(1)} = \frac{3 E^x \sin\left(\frac{\beta \sqrt{\Delta} K_{\phi}}{\sqrt{E^x}}\right)^2}{\beta^2 {\Delta}} - \frac{ \sqrt{E^x} K_{\phi} \sin\left( \frac{2\beta \sqrt{\Delta} K_{\phi}}{\sqrt{E^x}}\right)}{\beta \sqrt{\Delta}}  \ , \quad {f}^{(2)} = \frac{\sqrt{E^x}\sin\left( \frac{2\beta \sqrt{\Delta} K_{\phi}}{\sqrt{E^x}}\right)}{2 \beta \sqrt{\Delta}}. 
\end{eqnarray}
This effective constraint coincides with Tibrewala's polymer Hamiltonian \cite{Tibrewala:2012xb}, which can be derived from a covariant mimetic Lagrangian as shown in \cite{BenAchour:2017ivq}. Corollary \ref{cor:C_and_dynamics} can be applied to this model, which lets us write the equations of motion after imposing LTB as
\be\label{eom_ltb_case2}
\partial_t E^x &=&  \frac{{E^x} \sin \left(\frac{2 \a K_{\phi}}{\sqrt{{E^x}}}\right)}{\a}\\
\partial_t K_{\phi} &=& \frac{-3 {E^x} \sin ^2\left(\frac{\a K_{\phi}}{\sqrt{{E^x}}}\right)+\a \sqrt{{E^x}} K_{\phi} \sin \left(\frac{2 \a K_{\phi}}{\sqrt{{E^x}}}\right)}{2 \a^2 \sqrt{{E^x}}}  - \frac{\ltbf^2-1}{2 \sqrt{E^x}}.
\ee 
In order to simplify the equations of motion, we can introduce the new set variables 
\begin{equation*}
b = \frac{K_{\phi}}{\sqrt{E^x}}, \qquad v ={E^x}^{3/2}
\end{equation*}
since then we have
\begin{eqnarray}
\label{reduced equations of motion}
\partial_t v= \frac{3 v \sin \left(2 \a  b \right)}{2 \a },\quad \partial_t b =- \frac{1}{2} \left(\frac{\ltbf(x)^2-1}{v^{2/3}}+\frac{3 \sin ^2\left(\a b\right)}{\a^2}\right).
\end{eqnarray}
The effective Hamiltonian takes the form
\be\label{eq:Ham_Tibre_mubar}
C^{\Delta}(x) = \frac{1}{2G} \left( \frac{\ltbf'  }{\ltbf} + \partial_x \right) \left( \frac{ v  \sin(\beta \sqrt{\Delta} b)^2}{\beta^2 {\Delta} \ltbf} - \frac{v^{1/3}}{\ltbf} \left({\ltbf^2-1}\right) \right)(x)
\ee 
with the Dirac bracket given by
\begin{eqnarray}\label{eq:cano_Tibre_bv}
    \poissonbracket{b(x)}{v(y)}_D = - 3G \frac{\ltbf^2}{\ltbf'} \delta(x-y)\,.
\end{eqnarray}
We note that for this $\bar{\mu}$-scheme model Birkhoff's theorem can be applied, namely the vacuum solution is uniquely given as a static and asymptotically flat solution.
The reason for this is that the theory is scale invariant, thus all solutions of the ordinary differential equations in vacuum are characterized by a single integration parameter which corresponds to the black hole mass. The asymptotic flatness property then can be easily seen from the fact that the model admits a symmetric bounce and when $E^x \to \infty$ we have $\frac{K_{\phi}}{\sqrt{E^x}} \to 0$, and hence the classical theory is approached. We leave the detailed analysis of the dynamics of this model to \cite{noShocksPaper}.

\section{Conclusions}
\label{sec:Concl}
In this work we presented a general analysis on how compatible LTB conditions can be derived for spherically symmetric effective models that allow to embed a (generalized) LTB model into these kind of effective models. We extend the seminal work in \cite{Bojowald:2008ja,Bojowald:2009ih} in several directions. At the classical level we show in section \ref{sec:ClassMod} that only in the case of non-marginally bound LTB models, the LTB condition can be treated as a gauge fixing condition for the spatial diffeomorphism constraint. In the marginally bound model, the LTB condition is an additional first class constraint, yielding a model that has no physical degrees of freedom. At first this might look a bit counter-intuitive. However, taking into account the relation between the marginally bound model and the stationary vacuum solution, this becomes more obvious. We did not discuss that relation in detail in this work but refer to our separate work in \cite{GenBirkhoff}, where we also show that one can formulate a Birkhoff-like theorem for these kind of effective models under consideration. 
~\\
Next in section \ref{sec:EffMod} we discuss how generalized LTB models can be embedded in effective spherically symmetric models. As a given effective model comes with a chosen polymerization, first we list properties the effective models we consider should satisfy. This includes a function preserving polymerization in order to mimic a function preserving quantization. Further, we consider no polymerization of the spatial diffeomorphism constraint. The main motivation for this is that we also consider the algebra of the geometric contribution to the Hamiltonian constraint and the spatial diffeomorphism. At the classical level this forms a closed algebra. If we require this also for the corresponding effective model it will constrain the possible polymerization. Such an investigation significantly simplifies if the spatial diffeomorphism constraint is not polymerized. Interestingly, for the effective models considered here it turns out that in this class of models a polymerization of the extrinsic curvature $K_x$ cannot be present if the algebra is required to close. The class of models that satisfy this requirement is covered by Lemma \ref{lemma:f1_CC}. 
~\\
~\\
The strategy we follow to embed (generalized) LTB solutions differs partly from the one in the work in \cite{Bojowald:2008ja,Bojowald:2009ih}, where suitable LTB conditions are derived using an effective Hamiltonian. Here, we perform our derivation at the level of the equations of motion. The reason for choosing this approach is that the equations
will have a much simpler structure and so algebraic manipulations and numerical analysis are less complicated. An advantage is that this allows to consider holonomy and inverse triad corrections simultaneously, extending the results of \cite{Bojowald:2008ja,Bojowald:2009ih}, where these two kinds of corrections were only investigated separately. For this purpose we define the notion of a compatible LTB condition in Definition \ref{def:comp_LTB}. We consider compatible LTB conditions not only in the class of models covered by Lemma \ref{lemma:f1_CC}, i.e. effective models where the polymerized gravitational contribution to the Hamiltonian constraint is a conserved quantity, but also allow polymerizations that do not satisfy the assumptions stated in that Lemma. In Lemma \ref{thm:marginal} we present the form of the polymerization and the compatible LTB condition for the marginally bound case. Our results show that there exists an interplay between the polymerization of $K_x$ and the one of the compatible LTB condition. This has the consequence, that for a generically chosen polymerization of $K_x$, there might exist no compatible LTB condition at all. We further derive the general form of the polymerization and compatible LTB condition for the non-marginally bound case, see Lemma \ref{thm:non_marginally bound case }. 
~\\
~\\
As this is useful for our applications in section \ref{sec:examples} we further discuss four Corollaries. The first one, Corollary \ref{cor:gEx_m}, considers the marginally bound case and models that involve no polymerization of $\widetilde{K}_x=K_x/E^\phi$ for which a compatible LTB condition exist. Corollary \ref{cor:gEx_non_m} covers non-marginally bound models for which the LTB condition has a specific form. This carries over to a special form of the corresponding compatible LTB condition as well as restrictions on the possible polymerizations. While Corollary \ref{cor:gEx_m} and \ref{cor:gEx_non_m} not only consider models covered by Lemma \ref{lemma:f1_CC}, if restrict to those models, Corollary \ref{cor:fEx_clo} and Corollary \ref{cor:C_and_dynamics} give further insights. Corollary \ref{cor:fEx_clo} shows for this class of models a compatible LTB condition always exists. Moreover, it is shown that the classical LTB condition is a compatible one if additionally a certain relation among the inverse triad corrections is satisfied. This relation also encodes the special case when no inverse triad corrections are present, showing that there exist polymerized models with holonomy corrections only where the classical LTB condition is a compatible one. This is in contrast to the results in \cite{Bojowald:2008ja,Bojowald:2009ih} where for effective models always a non-trivial modification of the LTB condition is considered. As our application in section \ref{sec:example_1} to the models in \cite{Bojowald:2008ja,Bojowald:2009ih} demonstrates, this is due to the fact that the models in \cite{Bojowald:2008ja,Bojowald:2009ih} do not fall into the models covered by Lemma \ref{lemma:f1_CC}. Our results further show that the way contributions of the spin connections are treated is crucial and, in general, restricting the spin connection to its value in the LTB sector in the effective Hamiltonian, as it was done in \cite{Bojowald:2008ja}, does not yield  the correct result for the compatible LTB condition. Finally, Corollary \ref{cor:C_and_dynamics} considers models consistent with Lemma \ref{lemma:f1_CC} that in addition have a specific form of the compatible LTB condition considered in Corollary \ref{cor:gEx_m} and \ref{cor:gEx_non_m}. For these kind of models the generic form of the effective dynamics is presented in Corollary \ref{cor:C_and_dynamics} as well as its relation to the Dirac bracket for the non-marginally bound case that was introduced in the classical theory in section \ref{sec:ClassNonMarg}. Interestingly, for this class of models, the 
effective dynamics decouples completely along the radial direction. Such a decoupling was for instance assumed in the dust shell models in \cite{Giesel:2021dug,Kiefer:2019csi} and the results here support such an assumption. Corollary \ref{cor:C_and_dynamics} also covers the model we consider in our second application in \ref{sec:example_2} introduced in \cite{Tibrewala:2012xb}. It provides an example of an effective model with holonomy corrections for which the classical LTB condition is a compatible one. A more detailed analysis of the effective dynamics in the context of a dust collapse can be found in a separate paper \cite{noShocksPaper}. 
~\\
Since our study presented in this paper is also based on the condition that the class of effective models we consider satisfy a list of assumptions, it will be interesting to understand to what extent the formalism needs to be generalized if some of the assumptions are dropped. From the perspective of full LQG, it will be particularly interesting to see how a polymerized spatial diffeomorphism constraint can be incorporated into the analysis. We expect this to be more difficult to accomplish, as it is beyond the scope of this paper and we leave it to future work. Another direction in which assumptions can be relaxed is the kind of polymerizations that we considered in this analysis. For instance polymerizations of the form considered in the model in \cite{Alonso-Bardaji:2022ear} are not included in our work here. One reason for this is that these kind of polymerization cannot be linked to a mimetic gravity model, but there might be more general scalar-tensor theories that could be related to such a polymerization. We also leave the generalization as far as the choice of polymerizations are concerned as possible next steps. Another part that can be extended in future work is the detailed study of the dynamical aspects of the different classes of effective LTB models. A starting point for this is the computation of the induced Dirac brackets of the reduced system for given polymerized constraints and gauge fixings. From this one can then, after deriving the physical Hamiltonian, study the equation of motions of the system. 

\section*{Acknowledgements}
This work is supported by the DFG-NSF grants PHY-1912274 and 425333893 and NSF grant PHY-2110207.

\bibliographystyle{jhep.bst}
\bibliography{refs}

\appendix
\section{Computation of Dirac brackets of the classical non-marginal LTB system}\label{appendix:Diracbrackets_classical}
We want to start by computing the Poisson brackets of the constraints. The two main second class pairs are
\begin{align}\label{eq:1stPB}
       \poissonbracket{G_T(x)}{C^\mathrm{tot}(y)} &= \delta(x,y)\,\sqrt{1 + \frac{\lvert E^x \rvert}{(E^\phi)^2} (T')^2}\approx\delta(x,y)
 \end{align}
 and
 \begin{align}
 \label{eq:appendixpoissonGxdiffeo}    
        \poissonbracket{G_x(x)}{C_x^\mathrm{tot}(y)} &= \poissonbracket{\frac{{E^x}'}{2E^\phi}(x) - \ltbf(x)}{\frac{1}{G}(K_\phi' E^\phi - K_x {E^x}')(y)}  = \nonumber \\
        &= \frac{{E^x}'(x)E^\phi(y)}{2 E^\phi(x)^2}\,\frac{\mathrm{d}}{\mathrm{d}y}\delta(x,y) + \frac{  {E^x}'(y)}{2E^\phi(x)}\,\frac{\mathrm{d}}{\mathrm{d}x}\delta(x,y)\approx\nonumber \\
        &\approx \frac{{E^\phi}(y)}{{E^\phi}(x)}\qty(\ltbf(x) - \ltbf(y))\,\frac{\mathrm{d}}{\mathrm{d}y}\delta(x,y) = \ltbf'\,\delta(x,y)\,.
\end{align}
Note that $C^\mathrm{tot}$ and $C_x^\mathrm{tot}$ fulfill the standard hypersurface deformation algebra. The Poisson bracket
\begin{equation}\label{eq:appendixpoissonGxC}
    \poissonbracket{G_x(x)}{C^\mathrm{tot}\qty[N]} \approx N \frac{\sqrt{E^x}G}{{E^\phi}^2}  C_x + N' K_\phi \frac{\sqrt{E^x}}{ E^\phi} \eqqcolon \zeta (x)
\end{equation}
  is not vanishing as well. This will not be a problem since we are not interested in Dirac brackets of quantities involving the momentum of the clock field. We can see this even better when considering the Dirac matrix. The Dirac matrix $\mathcal{D}$ of the constraint vector $C^J := (G_T, G_x, C^\mathrm{tot}, C_x^\mathrm{tot})^{T}$ takes the form
\begin{align}
    \mathcal{D}^{I J}(x,y) = \poissonbracket{C^I(x)}{C^J(y)} \approx \delta(x,y) \Pmqty{ 0 & 0 & 1 & 0\\ 0 & 0 & \,\zeta & \,\,\,\ltbf'\,\,\, \\ -1 & -\zeta & 0 & 0 \\ 0 & -\ltbf' & 0 & 0}.
\end{align}
 Introducing the Dirac bracket
 \begin{equation}
     \poissonbracket{A(x)}{B(x)}_D \coloneqq \poissonbracket{A(x)}{B(x)} + \int \mathrm{d}z_1 \int \mathrm{d}z_2\, \poissonbracket{A(x)}{C^I(z_1)} \left(\mathcal{D}^{-1}\right)_{IJ}(z_1,z_2) \poissonbracket{C^J(z_2)}{B(y)},
 \end{equation}
we can finally compute the algebra of basic phase space variables
\begin{align*}
    \bullet \ \poissonbracket{K_x(x)}{E^x(y)}_D &= G \delta(x,y) + \int \mathrm{d}z_1 \int \mathrm{d}z_2\, \poissonbracket{K_x(x)}{G_x(z_1)} \left(\mathcal{D}^{-1}\right)_{24}(z_1,z_2) \poissonbracket{C^{\mathrm{tot}}_x(z_2)}{E^x(y)} \\
    &= G \qty(1 + \frac{\ltbf}{\ltbf'}(y)\pdv{y})\delta(x,y)\\
    \bullet \ \poissonbracket{K_\phi(x)}{E^\phi(y)}_D &= G \delta(x,y) + \int \mathrm{d}z_1 \int \mathrm{d}z_2\, \poissonbracket{K_\phi(x)}{G_x(z_1)} \left(\mathcal{D}^{-1}\right)_{24}(z_1,z_2) \poissonbracket{C^{\mathrm{tot}}_x(z_2)}{E^\phi(y)} \\
    &= G \qty(1 + \frac{\ltbf}{\ltbf'}(x)\pdv{x})\delta(x,y) \\
    \bullet \ \poissonbracket{K_x(x)}{K_\phi(y)}_D &= 0 + \int \mathrm{d}z_1 \int \mathrm{d}z_2\, \poissonbracket{K_x(x)}{G_x(z_1)} \left(\mathcal{D}^{-1}\right)_{24}(z_1,z_2) \poissonbracket{C^{\mathrm{tot}}_x(z_2)}{K_\phi(y)} \\
    &\quad \ \ \,   + \int \mathrm{d}z_1 \int \mathrm{d}z_2\, \poissonbracket{K_x(x)}{C^{\mathrm{tot}}_x(z_1)} \left(\mathcal{D}^{-1}\right)_{42}(z_1,z_2) \poissonbracket{G_x(z_2)}{K_\phi(y)}=\\
    &= G \qty(\frac{{K_\phi}'}{2} -  K_x \ltbf) \frac{1}{{E^\phi}\ltbf'}\pdv{y}\delta(x,y)\approx G^2 \frac{C_x}{2 {E^\phi}^2 \ltbf'}\pdv{y}\delta(x,y)\\
    \bullet \ \poissonbracket{E^x(x)}{E^\phi(y)}_D &= 0 \\
    \bullet \ \poissonbracket{K_x(x)}{E^\phi(y)}_D &= \int \mathrm{d}z_1 \int \mathrm{d}z_2\, \poissonbracket{K_x(x)}{G_x(z_1)} \left(\mathcal{D}^{-1}\right)_{24}(z_1,z_2) \poissonbracket{C^{\mathrm{tot}}_x(z_2)}{E^\phi(y)} \\
    &= \frac{G}{2}\pdv{x}\qty( \frac{1}{\ltbf'(x)}\pdv{x}\delta(x,y))\\
   \bullet \  \poissonbracket{K_\phi(x)}{E^x(y)}_D &= \int \mathrm{d}z_1 \int \mathrm{d}z_2\, \poissonbracket{K_\phi(x)}{G_x(z_1)} \left(\mathcal{D}^{-1}\right)_{24}(z_1,z_2) \poissonbracket{C^{\mathrm{tot}}_x(z_2)}{E^x(y)} \\
    &= -2 G \,\frac{\ltbf^2}{\ltbf'}(x)\delta(x,y) \,.
\end{align*}

\section{Details of \texorpdfstring{$\{ C^{\Delta}[N_1], C^{\Delta} [N_2]\} $}{scalar density bracket} }\label{app:bracketpolyHamiltonconstraint}
The bracket of polymerized gravitational contribution to the scalar density with itself can be written as

\be
\{ C^{\Delta}[N_1], C^{\Delta} [N_2]\} = (A {{E^x}}' + B K_x' + C K_{\phi}'+ D {{E^{\phi}}}')[N_2 N_1'-N_1 N_2'] \;,
\ee 
where we defined
\begin{eqnarray}
    A &\coloneqq& \frac{1}{{8 {E^x} {E^{\phi}}} } \Big( 4 {E^x} (1+f) (K_{\phi}-h_1 K_{\phi}+2 {E^x} Z)+K_{\phi} \big(-8 {E^x}^2 \partial_{E^x}f +2 {E^x} (K_{\phi}+4 {E^x} \widetilde{K}_x) \partial_{K_{\phi}}f \nonumber\\
    &&-((h_1-2) K_{\phi}+4 {E^x} h_1 Z) \partial_{\widetilde{K}_x}f-2 {E^x} (K_{\phi}+4 {E^x} \widetilde{K}_x) \partial_{\widetilde{K}_x} \partial_{E^x}f\big) \Big) \\
    B &=& - D \coloneqq -\frac{h_2 K_{\phi} \left(8 {E^x} \partial_{\widetilde{K}_x}f+(K_{\phi}+4 {E^x} \widetilde{K}_x) \partial_{\widetilde{K}_x}^2f\right)}{4 {E^{\phi}}^2} \\
    C&\coloneqq& -\frac{1}{4 {E^{\phi}}}\left(4 {E^x} \left(1+f+K_{\phi} \partial_{K_{\phi}}f\right)+2 (K_{\phi}+2 {E^x} Z) \partial_{\widetilde{K}_x}f+K_{\phi} (K_{\phi}+4 {E^x} \widetilde{K}_x) \partial_{\widetilde{K}_x} \partial_{K_{\phi}}f\right)\,.
\end{eqnarray}
To be in accordance with the hypersurface deformation algebra, above Poisson bracket has to be proportional to the spatial diffeomorphism constraint
\be
\poissonbracket{C^{\Delta}[N_1]}{C^{\Delta} [N_2]}\propto   C_x[N_2 N_1'-N_1 N_2'] \,.
\ee 
We can express $K^\prime_\phi\approx -\widetilde{K}_xE^{x\prime}$ with $\widetilde{K}_x:=\frac{K_x}{E^\phi}$ for solutions of diffeomorphism constraint $C_x$ which lets us write
\be\label{condition_diffeo_equal}
(A {{E^x}}' + B K_x' + C K_{\phi}'+ D {{E^{\phi}}}') \big|_{C_x = 0} = (A  - C  \widetilde{K}_x) {{E^x}}'+ B K_x' + D {{E^{\phi}}}'= 0 \,.
\ee
Due to the fact that no spatial derivatives appear in factor functions $A,B,C,D$, we can conclude
\be\label{new_condi}
B =D =0 \quad\text{and}\quad A - C \widetilde{K}_x = 0\,.
\ee 
The first condition $B =D =0$ in \eqref{new_condi} constrains the polymerization function $f(\widetilde{K}_x,K_{\phi},{{E^x}})$ to the following form
\begin{eqnarray}\label{eq:no_kx_f1_app}
    f = \frac{{f}^{(1)}(K_{\phi},E^x)- {f}^{(2)}(K_{\phi},E^x) K_{\phi}}{(K_{\phi}+4 {E^x} \widetilde{K}_x) K_{\phi}} + \frac{{f}^{(2)}(K_{\phi},E^x)}{ K_{\phi}} -1 \,.
\end{eqnarray}
Now further requiring $A - C \widetilde{K}_x =0$ together with \eqref{eq:no_kx_f1_app} leads to the following condition for the inverse triad corrections encoded in $h_1$ and $h_2$
\be\label{final_con_clo_app}
\frac{h_1 - 2 E^x \partial_{E^x} h_2}{h_2}= \frac{-4 {E^x} \partial_{E^x}{f}^{(2)}+  \partial_{K_{\phi}}{f}^{(1)}}{2 {f}^{(2)}} :=\text{Con}_{f}\,.
\ee 
If we go back to the generic ansatz for the polymerization function $f$ in \eqref{poly_f}, equation \eqref{eq:no_kx_f1_app} restricts polymerization to the form
\begin{eqnarray}
    (1+f)\qty(\frac{4 K_x K_{\phi}}{E^{\phi}} + \frac{K_{\phi}^2}{E^x}) \to \left( \frac{4 K_x {f}^{(2)}}{E^{\phi}} + \frac{{f}^{(1)}}{E^x} \right)\,.
\end{eqnarray}
As a result, the polymerization of $K_{x}$ is completely removed. Reinserting the solution for the polymerization function $f$ into $C^{\Delta}$ defined in \ref{eq:defpolyhamiltonianconstraint} and using \ref{final_con_clo_app}, the algebra now becomes
\be
\poissonbracket{C^{\Delta}[N_1]}{C^{\Delta}[N_2]} = \left(\big( \partial_{K_{\phi}} {f}^{(2)}\big)\frac{E^x }{(E^{\phi})^2}\,C_x\right) [N_1 N_2' - N_2 N_1']\,.
\ee 
This is identical to the classical result up to the deformation function $\partial_{K_{\phi}} {f}^{(2)}$ that depends on the chosen form of the polymerization function still encoded in the choice of ${f}^{(2)}$.

\section{Equations of motion}\label{app:eom}
The equations of motion generated by primary Hamiltonian $H_P^{\Delta}[N^x = 0]$ in \eqref{eq:defprimarypolyHamiltonian}, where the gravitational contribution to the Hamiltonian constraint is given by \eqref{eq:defpolyhamiltonianconstraint} and the polymerization function $f$ fulfills equation \eqref{eq:f_weight}, assume the form
\begin{eqnarray}\label{EoMeff}
    \partial_t E^x &=& \frac{K_{\phi} \left(4 {E^x} (1+f)+(4 {E^x} \widetilde{K}_x + K_{\phi}) \partial_{\widetilde{K}_x}f\right)}{2 \sqrt{{E^x}}} \\
    \partial_t E^{\phi} &=&\frac{E^{\phi} \left( 2 (1+f) (2 {E^x} \widetilde{K}_x+ K_{\phi})+K_{\phi} (4 {E^x} \widetilde{K}_x+ K_{\phi}) \partial_{K_{\phi}}f \right)}{2 \sqrt{{E^x}}}\\
    \partial_t K_x &=&\frac{1}{16 {E^x}^{3/2} {E^{\phi}}^2}\Big( 4 {E^x} {E^x}' {E^{\phi}}' \left(h_1-2 %
    \right)+4 {E^{\phi}}^3 \big( h_1-4 {E^x} (1+f) \widetilde{K}_x K_{\phi}\\
    &&+ (1+f) K_{\phi}^2-2 {E^x} K_{\phi} (4 {E^x} \widetilde{K}_x + K_{\phi}) \partial_{E^x}f-2 {E^x} \partial_{E^x} h_1 \big) \nonumber\\
    &&+4 {E^{\phi}} {E^x} {E^x}'' \left(-h_1+2 %
    \right)+{E^{\phi}} \left({E^x}'\right)^2 \left(h_1-2 {E^x} \partial_{E^x} h_1-2%
    \right) \Big) \nonumber\\
    \partial_t K_{\phi}&=&\frac{-4 {E^{\phi}}^2 \left(h_1+(1+f) K_{\phi}^2+\widetilde{K}_x K_{\phi} (4 {E^x} \widetilde{K}_x +K_{\phi}) \partial_{\widetilde{K}_x}f\right)+\left({E^x}'\right)^2 \left(2-h_1 %
    \right)}{8 \sqrt{{E^x}} {E^{\phi}}^2}\,.
\end{eqnarray}
We can use the relations in \eqref{eq:Ep_Kx_ltb} to rewrite the above equations as
\begin{eqnarray}
    \partial_t K_{\phi} = \frac{{g}^2(2-h_1) -h_1}{2 \sqrt{E^x}} + \mathcal{F}_{K_\phi}(K_{\phi}, E^x, \widetilde{K}_x)\ , \qquad
    \partial_t E^x = \mathcal{F}_{E^x}(K_{\phi}, E^x, \widetilde{K}_x)
\end{eqnarray}
where we defined
\begin{eqnarray}
     \mathcal{F}_{K_\phi}(K_{\phi}, E^x, \widetilde{K}_x)  \coloneqq- \frac{ (1+f) K_{\phi}^2+\widetilde{K}_x K_{\phi} (4 {E^x} \widetilde{K}_x +K_{\phi}) \partial_{\widetilde{K}_x}f%
    }{2 \sqrt{{E^x}}} \\
    \mathcal{F}_{E^x}(K_{\phi}, E^x, \widetilde{K}_x) \coloneqq  \frac{K_{\phi} \left(4 {E^x} (1+f)+(4 {E^x} \widetilde{K}_x + K_{\phi}) \partial_{\widetilde{K}_x}f\right)}{2 \sqrt{{E^x}}}\,.
\end{eqnarray}
Finally we can derive a differential equation for ${g} = \frac{  {E^x}'}{2 E^{\phi}}$ using above equations of motion
\begin{eqnarray}\label{eqapp:g1_eom}
        \partial_t {g} =&& {{g}} \left( \frac{ ( \partial_t {E^x})'}{{E^x}'} -\frac{ \partial_t E^{\phi}}{E^{\phi}} \right) 
        = {g} \mathcal{F}_{{g}}(K_{\phi}, E^x, \partial_x K_{\phi}, \partial_x {E^x}, \partial_x^2 K_{\phi}, \partial_x^2 {E^x})
\end{eqnarray}
with
\begin{eqnarray}
           \mathcal{F}_{{g}}=&&\frac{{g}}{4 {E^x}' {E^x}^{3/2}} \Bigg[-{E^x}' K_{\phi}^2 \partial_{\widetilde{K}_x}f+2 {E^x} K_{\phi} \bigg(4 {E^x}' \widetilde{K}_x \partial_{\widetilde{K}_x}f+{\widetilde{K}_x}' K_{\phi} \partial_{\widetilde{K}_x}^2f\\
        &&+{E^x}' K_{\phi} \left(-\partial_{K_{\phi}}f+\partial_{\widetilde{K}_x} \partial_{E^x}f+\widetilde{K}_x \partial_{\widetilde{K}_x} \partial_{K_{\phi}}f\right)\bigg)+8 {E^x}^2 \bigg({E^x}' K_{\phi} \partial_{E^x}f \nonumber \\
        &&+\left(2 {\widetilde{K}_x}' K_{\phi}+{E^x}' \widetilde{K}_x^2\right) \partial_{\widetilde{K}_x}f +K_{\phi} \widetilde{K}_x \left({E^x}' \partial_{\widetilde{K}_x} \partial_{E^x}f+{E^x}' \widetilde{K}_x \partial_{\widetilde{K}_x} \partial_{K_{\phi}}f+{\widetilde{K}_x}' \partial_{\widetilde{K}_x}^2f\right)\bigg)\Bigg]\,.  \nonumber
\end{eqnarray}

\section{LTB condition under general polymerization}\label{app:ltb_general}
We want to start with the most general ansatz ${g} = {g}({K_{\phi}}',{E^x}', K_{\phi},E^x,\Xi)$ for a compatible LTB condition according to Lemma \ref{lemma:g1_first}. (In this lemma it is shown why the higher derivatives of $K_\phi$ and $E^x$ drop out). Note that already at this level ${g} = {g}(\widetilde{K}_x, K_{\phi},E^x,\Xi)$ is the only possible possible ansatz to make ${g}$ a scalar density with weight zero. The conditions for $n=1$ are given by
\begin{eqnarray}
\label{eqapp:fc1}
   F_{C_1} = &&-\frac{1}{2 ({E^x}')^2 \sqrt{{E^x}}} \bigg( K_{\phi} \left(8 {E^x} \partial_{\widetilde{K}_x}f+(K_{\phi}+4 {E^x} \widetilde{K}_x) \partial_{\widetilde{K}_x}^2f\right) \left({g}- {E^x}' \partial_{{E^x}'} {g} \right)\\
    &&- {E^x}' \left( K_{\phi} \widetilde{K}_x \left(8 {E^x} \partial_{\widetilde{K}_x}f+(K_{\phi}+4 {E^x} \widetilde{K}_x) \partial_{\widetilde{K}_x}^2f\right) - 2  ({E^x}')^2 (2 h_2 + 4 E^x \partial_{E^x} h_2 - h_1) {g}  \partial_{{K_{\phi}}'} {g} \right) \partial_{{K_{\phi}}'} {g}\bigg) \nonumber \\
   \label{eqapp:fd1}
    F_{D_1} = &&\frac{1}{2 ({E^x}')^2 \sqrt{{E^x}}} \bigg(  K_{\phi} \widetilde{K}_x \left(8 {E^x} \partial_{\widetilde{K}_x}f+(K_{\phi}+4 {E^x} \widetilde{K}_x) \partial_{\widetilde{K}_x}^2f\right) \left( {g}-{E^x}' \partial_{{E^x}'} {g}\right) \\
    &&-{E^x}' \left(K_{\phi} (\widetilde{K}_x)^2 \left(8 {E^x} \partial_{\widetilde{K}_x}f+(K_{\phi}+4 {E^x} \widetilde{K}_x) \partial_{\widetilde{K}_x}^2f\right) - 2 {E^x}' (2 h_2 + 4 E^x \partial_{E^x} h_2 -h_1)  {g} \partial_{{E^x}'} {g}\right) \partial_{{K_{\phi}}'} {g}  \bigg) \,.\nonumber
\end{eqnarray} 
We can add them together to get the relation
\begin{eqnarray}
       0 = F_{D_1} +  \widetilde{K}_x F_{C_1} =&&\frac{(2 h_2 + 4 E^x \partial_{E^x} h_2 -h_1)  {g} \partial_{{K_{\phi}}'} {g} \left(\partial_{{E^x}'} {g}+ \widetilde{K}_x \partial_{{K_{\phi}}'} {g}\right)}{\sqrt{{E^x}}}\,,
\end{eqnarray}
which implies that ${g}$ can be written as
\begin{eqnarray}
    {g} = {g}_{\Delta}({E^x}', K_{\phi},E^x,\Xi) \, \quad \text{or} \, \quad {g} = {g}_{\Delta}(\widetilde{K}_{x}, K_{\phi},E^x,\Xi)\,. 
\end{eqnarray}
Plugging the first type of solutions for ${g}$ back in equation \eqref{eqapp:fc1}, we can deduce due to $\partial_{{K'_{\phi}}} {g}= 0 $ that
\begin{eqnarray}
    {g} = {g}_{\Delta}( K_{\phi},E^x,\Xi){E^x}'
\end{eqnarray}
which can not have the correct classical limit. We note the special case when 
\begin{eqnarray}
    8 {E^x} \partial_{\widetilde{K}_x}f+(K_{\phi}+4 {E^x} \widetilde{K}_x) \partial_{\widetilde{K}_x}^2f =0\,,
\end{eqnarray}
since then ${g} = {g}_{\Delta}({E^x}', K_{\phi},E^x,\ltbf)$ seems to be allowed. However, the equation of motion $\partial_t {g} = {{g}} \mathcal{F}_{{g}}$ in such case shows that $\partial_{{E^x}'} {g}_{\Delta}({E^x}')=0$ and thus completely removes the dependence on ${E^x}'$. The reason behind this is that the equations of motion transform covariantly under $C_x$, which requires ${g}_{\Delta}$ to have density weight zero. In summary, ${g}$ can only take the form $ {g} = {g}_{\Delta}(\tilde{K}_{x}, K_{\phi},E^x,\ltbf) $ and $F_{C_1}$ further implies that
\begin{eqnarray}
    2 \left(2 h_2 + 4 E^x \partial_{E^x} h_2 - h_1 \right) \left(\partial_{\widetilde{K}_x}{g}_{\Delta}\right)^2 = 8 {E^x} K_{\phi} \partial_{\widetilde{K}_x}f + K_{\phi} (K_{\phi} + 4 {E^x}\widetilde{K}_x ) \partial_{\widetilde{K}_x}^2f\,.
\end{eqnarray}
Now substitute the ansatz $ {g} = {g}_{\Delta}(\tilde{K}_{x}, K_{\phi},E^x,\ltbf) $ back into \eqref{eqapp:g1_eom}, we obtain
\begin{eqnarray}\label{eqapp:LTB_gene}
    0=\frac{(2 h_2 + 4 E^x \partial_{E^x} h_2-h_1){g}}{\partial_x E^x} \partial_{\ltbf} {g} \partial_{\widetilde{K}_x} {g} \ltbf' + \frac{1}{4 (E^x)^{3/2}}\widetilde{F}_{C_0}(\widetilde{K}_x,K_{\phi},E^x,{g},\partial {g},f,\partial f)
\end{eqnarray}
with
\begin{equation}\label{eqapp:def_tfc0}
    \begin{split}
    \widetilde{F}_{C_0}=& \partial_{\widetilde{K}_x}f \Big[{g} \left(K_{\phi}^2-8 {E^x} K_{\phi} \widetilde{K}_x-8 {E^x}^2 \widetilde{K}_x^2\right)+2 {E^x} K_{\phi} (K_{\phi}+4 {E^x} \widetilde{K}_x) \left(\partial_{E^x}{g}+\widetilde{K}_x \partial_{K_{\phi}}{g}\right) \Big]\\
    &+ \Big[ {g}^2 (h_1-2 h_2-2 {E^x} h_1'+8 {E^x} \left(h_2'+{E^x} h_2''\right))
    +(1+f) \left(K_{\phi}^2-8 {E^x} K_{\phi} \widetilde{K}_x-8 {E^x}^2 \widetilde{K}_x^2\right)\\
    &-2 {E^x} \left(2 {g} (h_1-2 h_2) \left(\partial_{E^x}{g}+\widetilde{K}_x \partial_{K_{\phi}}{g}\right)+K_{\phi} (K_{\phi}+4 {E^x} \widetilde{K}_x) \left(\partial_{E^x}f+\widetilde{K}_x \partial_{K_{\phi}}f\right)\right)\\
    &+h_1-2 {E^x} h_1'+8 {E^x} {g} h_2' \left(2 {E^x} \left(\partial_{E^x}{g}+\widetilde{K}_x \partial_{K_{\phi}}{g}\right)\right)
    \Big]\partial_{\widetilde{K}_x}{g} \\
    &-2 {E^x} {g} K_{\phi} (K_{\phi}+4 {E^x} \widetilde{K}_x) \partial_{\widetilde{K}_x}\partial_{E_x}f -2 {E^x} {g} K_{\phi} \widetilde{K}_x (K_{\phi}+4 {E^x} \widetilde{K}_x)\partial_{\widetilde{K}_x}\partial_{K_{\phi}}f \\
    &+2 {E^x} \Big[{g} K_{\phi}( K_{\phi}\partial_{K_{\phi}}f - 4 E^x \partial_{E^x}f)-\left(h_1+{g}^2 (h_1-2 h_2 - 4 E^x h_2') +(1+f) K_{\phi}^2\right) \partial_{K_{\phi}}{g}\\
    &+4 {E^x} (1+f) K_{\phi} \partial_{E^x}{g}- {g} K_{\phi} \left(K_{\phi} + 4 {E^x}\right) \widetilde{K}_x \left(\partial_{\widetilde{K}_x} \partial_{E^x}f+\widetilde{K}_x \partial_{\widetilde{K}_x} \partial_{K_{\phi}}f\right)\Big]\,.
    \end{split}
\end{equation}
Note that when we require 
\begin{eqnarray}
    8 {E^x} \partial_{\widetilde{K}_x}f+(K_{\phi}+4 {E^x} \widetilde{K}_x) \partial_{\widetilde{K}_x}^2f = \partial_{\widetilde{K}_x} {g} =0 \, 
\end{eqnarray}
namely $f$ is given by \eqref{eq:no_kx_f1}, then above equation becomes
\begin{eqnarray}\label{eq:G_expan_ltb}
    \tilde{F}_{C_0} = \mathcal{G}_1+ \mathcal{G}_2
\end{eqnarray}
with
\begin{equation}\label{eqapp:G12_expan_explicit}
    \begin{split}
          \mathcal{G}_1=& 2 {f}^{(2)} \left( {g}_{\Delta} -2 E^x \partial_{E^x} {g}_{\Delta}\right) +{g}_{\Delta} \left(4 E^x \partial_{E^x} {f}^{(2)} - \partial_{K_{\phi}}{f}^{(1)}\right)
        -(h_1+ {f}^{(1)}) \partial_{K_{\phi}} {g}_{\Delta}\\
        \mathcal{G}_2 =&({g}_{\Delta})^2 \left(2 h_2 + 4 E^x \partial_{E^x} h_2 -h_1\right) \partial_{K_{\phi}} {g}_{\Delta}\,.
    \end{split}
\end{equation}

\end{document}